\documentclass[journal,twoside,web]{ieeecolor}
\usepackage{generic}
\usepackage{cite}
\def\BibTeX{{\rm B\kern-.05em{\sc i\kern-.025em b}\kern-.08em
    T\kern-.1667em\lower.7ex\hbox{E}\kern-.125emX}}
\usepackage[dvipsnames]{xcolor}
\usepackage{amsmath,amssymb,amsfonts}
\usepackage{algorithmic}
\usepackage{graphicx}
\usepackage{textcomp}

\usepackage{upgreek}
\usepackage{float}
\usepackage{enumerate}
 
\usepackage{amsthm}
\usepackage{cleveref}
\newtheorem{lem}{Lemma}
\crefname{lem}{Lemma}{Lemmas}
\newtheorem{thm}{Theorem}
\crefname{thm}{Theorem}{Theorems}
 
\crefname{cor}{Corollary}{Corollaries}
\newtheorem{prop}{Proposition}
\crefname{prop}{Proposition}{Propositions}
\theoremstyle{definition}
\newtheorem{rmk}{Remark}
\crefname{rmk}{Remark}{Remarks}
\newtheorem{defn}{Definition}
\crefname{defn}{Definition}{Definitions}
\newtheorem{assum}{Assumption}
\crefname{assum}{Assumption}{Assumptions}
\crefname{appendix}{Appendix}{Appendices}
\crefname{section}{Section}{Sections}
\crefname{table}{Table}{Tables}
\newtheorem{ex}{Example}
\crefname{ex}{Example}{Examples}
\newcommand\xqed[1]{
	\leavevmode\unskip\penalty9999 \hbox{}\nobreak\hfill\quad\hbox{#1}}
\newcommand\eoe{\xqed{$\triangle$}}

\def\usetikz{0} 
\if\usetikz1
	\usepackage{tikz}
	\usepackage{pgfplots}
	\usepackage[nooldvoltagedirection]{circuitikz} 
	\pgfplotsset{compat=1.14}
	\usepgfplotslibrary{groupplots}
	\usetikzlibrary{positioning}
	\usetikzlibrary{external} 
	\tikzset{
	    png export/.style={
			external/system call/.add={}
	        && convert -density 600 -transparent white "\image.pdf" "\image.png", 
	        /pgf/images/external info,
	        /pgf/images/include external/.code={%
	            \includegraphics
	                [width=\pgfexternalwidth,height=\pgfexternalheight]
	                {##1.png}%
	        },
		} 
	}
\fi

\newcommand{\setreal}{\mathbb{R}}

\newcommand{\setnat}{\mathbb{N}}

\newcommand{\transp}{\mathsf{T}}

\newcommand{\taumin}{\underline{\tau}}
\newcommand{\taumax}{\overline{\tau}}

\DeclareMathOperator*{\diag}{diag}
\DeclareMathOperator*{\card}{card}

\newcommand{\estnot}[1]{{\hat{#1}}}
\newcommand{\yest}{\estnot{v}}
\newcommand{\west}{\estnot{w}}
\newcommand{\thetaest}{\estnot{\theta}}

\newcommand{\vest}{\estnot{v}}
\newcommand{\mest}{\estnot{m}}
\newcommand{\hest}{\estnot{h}}
\newcommand{\xest}{\estnot{x}}

\newcommand{\maxcondest}{\estnot{\mu}}

\newcommand{\parvest}{\estnot{\parv}}
\newcommand{\parwest}{\estnot{\parw}}

\newcommand{\virt}[1]{{\tilde{#1}}}
\newcommand{\xvirt}{\virt{x}}
\newcommand{\vvirt}{\virt{v}}
\newcommand{\wvirt}{\virt{w}}
\newcommand{\parvvirt}{\virt{\parv}}
\newcommand{\parwvirt}{\virt{\parw}}
\newcommand{\fvvirt}{\virt{f}}
\newcommand{\fwvirt}{\virt{g}}

\newcommand{\Pvirt}{M}

\newcommand{\ratevirt}{\lambda}

\newcommand{\parv}{\theta}
\newcommand{\Parv}{\Theta}
\newcommand{\parw}{\eta}
\newcommand{\Parw}{H}
\newcommand{\bv}{a}
\newcommand{\bw}{b}

\newcommand{\Pv}{M_{v}}
\newcommand{\ratev}{\lambda_v}

\newcommand{\Pint}{M_{w}}

\newcommand{\rateint}{\lambda_w}

\newcommand{\vhalf}{\rho}
\newcommand{\slope}{\kappa}
\newcommand{\mean}{\zeta}
\newcommand{\std}{\chi}

\newcommand{\col}{\mathrm{col}}

\newcommand{\gain}{\gamma}

\newcommand{\vmax}{\overline{v}}
\newcommand{\vmin}{\underline{v}}
\newcommand{\umax}{\overline{u}}

\newcommand{\Pmin}{\underline{p}}
\newcommand{\Pmax}{\overline{p}}
\newcommand{\Mmin}{\underline{m}}
\newcommand{\Mmax}{\overline{m}}

\newcommand{\sat}{\varsigma}

\newcommand{\maxcond}{\mu}
\newcommand{\maxcondint}{\maxcond}

\newcommand{\nernst}{\nu}
\newcommand{\nernstint}{\nernst}

\newcommand{\nv}{{n_v}}
\newcommand{\nw}{{n_w}}

\newcommand{\ny}{{n_v}}

\newcommand{\nparv}{{n_{\theta_v}}}
\newcommand{\nparw}{{n_{\parw}}}

\newcommand{\Iapp}{u}
\newcommand{\Iint}{I}
\newcommand{\Isyn}{I}

\newcommand{\INa}{I_{\text{Na}}}
\newcommand{\IK}{I_{\text{K}}}
\newcommand{\IL}{I_{\text{L}}}

\newcommand{\Na}{{\textrm{Na}}}
\newcommand{\K}{{\textrm{K}}}
\newcommand{\Ca}{{\textrm{Ca}}}
\newcommand{\GABA}{{\textrm{G}}}
\newcommand{\Leak}{{\textrm{L}}}

\newcommand{\ion}{{\rm{ion}}}
\newcommand{\pre}{p}

\newcommand{\syn}{{\rm{syn}}}
\newcommand{\app}{{}}

\begin{document}

\title{Adaptive observers for biophysical neuronal circuits}

\author{Thiago B. Burghi and Rodolphe Sepulchre
\thanks{Submitted to IEEE Transactions
on Automatic Control.
The research leading to these results has received
funding from the European Research Council under the Advanced ERC
Grant Agreement Switchlet n.670645.}
\thanks{
Thiago B. Burghi and Rodolphe Sepulchre are with the Department of
Engineering, Control Group, University of Cambridge, CB2 1PZ, UK
(e-mails: tbb29@cam.ac.uk, r.sepulchre@eng.cam.ac.uk)}}

\maketitle


\begin{abstract}

This paper presents adaptive observers for online state and parameter estimation of a class of nonlinear systems motivated by biophysical models of neuronal circuits. We first present a linear-in-the-parameters design that solves a classical recursive least squares problem. Then, building on this simple design, we present an augmented adaptive observer for models with a nonlinearly parameterized internal dynamics, the parameters of which we interpret as structured uncertainty. We present a convergence and robustness analysis based on contraction theory, and illustrate the potential of the approach in neurophysiological applications by means of numerical simulations.
\end{abstract}

\section{Introduction}

With the development and refinement of neural recording technology, controlling the nervous system at the cellular scale may soon become possible. Techniques such as voltage imaging \cite{knopfel_optical_2019} promise to deliver simultaneous subthreshold recordings of large biological neural networks, opening up new possibilities for the design of brain-machine interfaces \cite{nicolas-alonso_brain_2012}. But while large-scale technologies are still maturing, closed-loop control of small living neuronal circuits has been a reality since the development of the \textit{dynamic clamp} \cite{sharp_dynamic_1993} electrophysiology technique. Even though the control of such small circuits is not yet done in a systematic fashion, it has enabled important scientific discoveries related to the electrochemical process of neuromodulation  \cite{marder_neuromodulation_2014}.

The systematic control of small neural circuits is an open problem \cite{drion_neuronal_2015} that will only become more challenging as the scale of the circuits is increased. The main bottleneck is the ever changing nature of living neurons \cite{sorrell_brainmachine_2021}, which implies that any model-based approach to neuronal control must consider online estimation methods. Any such estimation method should deal with the spiking nature of electrophysiological signals, and the consequent nonlinearity of state-space neuronal models \cite{izhikevich_dynamical_2007}. In particular, \textit{conductance-based models}, introduced in the seminal work \cite{hodgkin_quantitative_1952}, have a large number of uncertain parameters and unmeasured states, and dealing with this issue has been an important modelling challenge \cite{hille_ionic_1984}.

The question of estimating conductance-based neuronal models from input-output data has mostly been approached with \textit{offline} algorithms and \textit{output-error}   \cite{ljung_system_1999} model structures, see e.g. 
\cite{druckmann_novel_2007,meliza_estimating_2014,nogaret_automatic_2016}.
However, since the neuronal dynamics lack the fading memory property that is essential for performing output-error estimation
\cite{ljung_convergence_1978,ljung_system_1999},
such methods lead to difficult optimization problems with nonsmooth cost functions
\cite{abarbanel_estimation_2008,ribeiro_smoothness_2020}; as a consequence, the use of such methods in adaptive schemes is precluded. To deal with these difficulties, some authors have exploited the assumption that the only parameters to be estimated are a neuron's \textit{maximal conductances} (including synaptic weights), while other parameters related to ion channel properties can be assumed known. In this case, the neuronal model structure becomes linear-in-the-parameters \cite{huys_efficient_2006,narayanan_biophysically_2019,burghi_feedback_2021}. An important question related to such approaches is the effect of ion channel model uncertainty.

In this paper, we address the problem of \textit{online} estimation of single-neuron and neural network conductance-based models. Our modelling framework acknowledges the linear parametrization of maximal conductances, which are key players in the neuromodulation of neuronal behaviours, from the single-cell to the network scale \cite{marder_neuromodulation_2014,drion_neuronal_2015,drion_cellular_2019}. At the same time, we highlight the important issue of uncertainty in ion channel models, which define the internal dynamics of a neuron and its synapses. Our first contribution is the design and analysis of a globally convergent adaptive observer based on the classical recursive least squares (RLS) method, which assumes a linear-in-the-parameter neuronal output dynamics and a known nonlinear internal dynamics. Building on that design, we then propose an augmented adaptive observer capable of estimating parametric (structured) uncertainty in a nonlinearly parameterized neuronal internal dynamics. 

The observers in this paper are aligned with the literature on nonlinear adaptive observers \cite{gauthier_simple_1992,marino_global_1992,farza_adaptive_2009}.   
Our approach is however closer to linear observer design \cite{zhang_adaptive_2001} than to nonlinear observer design, since, instead of relying on particular state space observer normal forms \cite{krstic_nonlinear_1995,besancon_remarks_2000}, we rather rely on contraction theory principles \cite{lohmiller_contraction_1998}.
Contraction analysis provides a framework reminiscent of the linear theory of adaptive control, as well as explicit convergence rates and robustness guarantees grounded in the concept of a \textit{virtual system} \cite{jouffroy_tutorial_2010,bonnabel_contraction_2015}. Contraction analysis has been a driving methodology in recent adaptive control research  \cite{lopez_adaptive_2021}, and the present work demonstrates its value for the design of adaptive systems in neuroscience.

The paper is organized as follows. The model structure assumptions and their application to conductance-based models are presented in \cref{sec:modeling}. In \cref{sec:no_uncertainty}, the online estimation problem for a simplified linear-in-the-parameters model structure is studied, and a globally convergent adaptive observer is presented. In \cref{sec:main_result}, parametric nonlinear uncertainty in the internal dynamics is introduced, and we present an augmented adaptive observer to solve the estimation problem; we also discuss the effects of measurement noise and unstructured uncertainty. In \cref{sec:numeric}, we illustrate the performance of the adaptive observers and discuss the potential of the approach in neurophysiology.

\subsection{Notation}

For a finite-dimensional vector $x$, we write $n_x := \dim(x)$.
For two column vectors $x$ and $y$, we write 
$\col(x,y) := (x^\transp,y^\transp)^\transp$.
For a matrix $A \in \setreal^{n\times m}$, $\|A\|$
denotes the spectral norm (the largest singular value
of $A$). 
For a vector $x \in \setreal^{n_x}$ and a symmetric
matrix $P \in \setreal^{{n_x}\times {n_x}}$, we write 
$\|x\|_P^2 := x^\transp P x$, and $\|x\| := \|x\|_I$
with $I$ the identity matrix.
For a vector-valued
function $f:\setreal^{n_1} \times \setreal^{n_2} \to 
\setreal^m$, we write $\partial_x f(x,y) \in 
\setreal^{m \times n_1}$ for the Jacobian of $f(x,y)$
with respect to $x$. 
We write $A \succeq B$ ($A \succ B$) if 
$A-B$ is a positive-semidefinite 
(positive-definite) matrix. This paper often uses the formalism of contraction analysis \cite{lohmiller_contraction_1998}, which is briefly recalled in \cref{sec:contraction_analysis}.

\section{System class}
\label{sec:modeling}

This section introduces and motivates the model assumptions of the paper. \cref{sec:model_structure} defines the basic model structure and states our main assumptions. \cref{sec:cb_models} then shows how the model structure and the assumptions are motivated by our main application: the conductance-based model of a neuron. Finally, \cref{sec:cb_networks} shows that the assumptions extend from single neurons to models of neuronal networks.

\subsection{Problem statement}
\label{sec:model_structure} 

This paper considers nonlinear state-space systems of the form
\begin{subequations}
\label{eq:true_system} 
	\begin{align}
		\label{eq:dv_true} 
		\dot{v} &= 
		\Phi(v,w,u) \parv + \bv(v,w,u) \\
		\label{eq:dw_true} 
		\dot{w} &= 
		A(v,\parw) w + \bw(v,\parw)
	\end{align}
where $v(t)\in\setreal^\ny$ is a measured output, $w(t) \in \setreal^\nw$ are unmeasured internal states, and 
$\parv \in \setreal^\nparv$ and $\parw \in \setreal^\nparw$ are parameter vectors. We call \eqref{eq:dv_true} the \textit{output dynamics}, and \eqref{eq:dw_true} the \textit{internal dynamics}. We assume that $\Phi$, $A$, $\bv$ and $\bw$ are continuously differentiable functions of the appropriate dimensions. 

The model structure \eqref{eq:true_system} is motivated by neuroscience applications discussed in \cref{sec:cb_models}. In those applications, the vector $\parv$ is unknown, while $\parw$ is not unknown but uncertain. Thus we work in the context of structured model uncertainty \cite{sastry_adaptive_2011}. Our aim is to design an adaptive observer to estimate $\theta$ and, if necessary, $\parw$. For that purpose, we regard the parameters as part of the state of the system. We will initially consider the constant model 
\begin{equation}
	\label{eq:dtheta_true}
	  	\dot{\parv} = 	0, \quad
	  	\dot{\parw} = 0,	
\end{equation}		
\end{subequations}
so that $\parv(t) = \parv(0)$ and $\parw(t) = \parw(0)$ for all $ t \ge 0$; later, we will study the case where such parameters are time-varying.

We now consider the main assumptions on the properties of \eqref{eq:true_system}. These assumptions will also be justified by the  applications in \cref{sec:cb_models}.
\begin{assum}
	\label{assum:true_invariant_set} 
	There exists a compact set $U$ such that $u(t) \in U$ for all $t \ge 0$.	
	Furthermore, there exists a compact convex set $V \times W \times \{\parv(0)\}\times\{\parw(0)\}$ which is positively invariant with respect to \eqref{eq:true_system}, uniformly in $u$ on $U$.
\end{assum}

\begin{assum}
\label{assum:int_dyn_contraction}
There exist a symmetric positive definite matrix $\Pint \succ 0$
	and a contraction rate $\rateint > 0$ such that 
	\begin{equation}
		\label{eq:Mint}
		A(v,\parw)^\transp \Pint 
		+ \Pint A(v,\parw)
		\preceq -\rateint \Pint
	\end{equation}
	for all $\{v,\parw\} \in \setreal^{\ny} \times \setreal^\nparw$. It is assumed that $	\|M_w\| = 1$ 	without loss of generality.
\end{assum}

\begin{rmk}
\label{rmk:saturation_w} 
When \cref{assum:true_invariant_set} holds, then without loss of generality we can assume that for all $v \in V$ and $u \in U$, the functions $\Phi(v,w,u)$ and $\bv(v,w,u)$ are globally Lipschitz and bounded in $w \in \setreal^{\nw}$. This is because we can replace $w$ by $\sat_w(w)$ in the arguments of those functions, where $\sat_w:\setreal^{\nw} \to W$ is a smooth saturation function such that $\sat_w(w) = w$ for all $w \in W$. Doing so does not change the dynamics of \eqref{eq:true_system} within the positively invariant set of \cref{assum:true_invariant_set}.
\end{rmk}

The reader will note that the system \eqref{eq:true_system} is not in the classical \textit{output-feedback canonical form} \cite{krstic_nonlinear_1995}, nor in any of the equivalent adaptive observer forms summarized by \cite{besancon_remarks_2000}. The system also does not fit the model structures addressed in the more recent adaptive observer literature, e.g., \cite{farza_adaptive_2009,tyukin_adaptive_2013}. 


\subsection{Conductance-based single-neuron model}
\label{sec:cb_models} 

Since the seminal work of Hodgkin and Huxley \cite{hodgkin_quantitative_1952}, the nonlinear electrical circuits known as conductance-based models have become the foundation of biophysical modelling in neurophysiology  \cite{izhikevich_dynamical_2007}.
We now show that any single-neuron conductance-based  model can be written in the form \eqref{eq:true_system} in such a way that \cref{assum:true_invariant_set,assum:int_dyn_contraction} are satisfied.


A circuit representation of the model is shown in Figure~\ref{fig:conductance_based}: a capacitor of capacitance $c>0$ in parallel with 
a \textit{leak current} $I_{\rm{L}}$ and a number of \textit{intrinsic ionic currents} $\Iint_\ion$. The \textit{input current} $u(t)\in\setreal$ represents the external current, injected with an intracellular electrode.
The capacitor voltage $v(t)\in\setreal$ modelling the neuronal membrane potential evolves according to Kirchhoff's law, 
\begin{equation}
	\label{eq:single_neuron_cb}
		c \, \dot{v} = 
		- I_{\rm{L}} 
		- \sum_{\ion \in \mathcal{I}} 
		\Iint_\ion
		+
		u,
\end{equation}
where $\mathcal{I}=\{\ion_1,\ion_2,\dotsc,\ion_{\card(\mathcal{I})}\}$ is the (finite) index set of intrinsic ionic currents. Each current in the circuit is ohmic in nature, but with a conductance that can be nonlinear and voltage-dependent. The leak current has a constant conductance and is given by
\begin{equation}
	\label{eq:leak_current} 
	I_{\rm{L}} = \maxcond_{\rm{L}} (v-\nernst_{\rm{L}}),
\end{equation}
with $\maxcond_\Leak>0$, and the intrinsic ionic currents are modelled by
\begin{subequations}
	\label{eq:current_cb}
	\begin{align}
		\label{eq:ion_currents} 
		I_\ion &= \maxcondint_\ion \,
		m_\ion^{p_\ion} \, h_\ion^{q_\ion} \, 
		(v - \nernstint_\ion) \\[.3em]
		\label{eq:activation} 
		\tau_{m_\ion}(v) 
		\dot{m}_\ion &=  
			-m_\ion+
			\sigma_{m_\ion}(v) 			
			\\[.3em]
		\label{eq:inactivation} 
		\tau_{h_\ion}(v) \dot{h}_\ion &= 
			-h_\ion +
			\sigma_{h_\ion}(v) 			
	\end{align}
\end{subequations}
The constants $\maxcondint_{\ion}>0$ and $\nu_{\ion}\in\setreal$ are called (intrinsic) \textit{maximal conductances} and \textit{reversal potentials}, respectively. The exponents $p_\ion$ and $q_\ion$ in \eqref{eq:ion_currents} are fixed natural numbers (including zero).
The static \textit{activation functions} 
$\sigma_{m_\ion}(v)$
and 
$\sigma_{h_\ion}(v)$,
and
\textit{time-constant functions} $\tau_{m_\ion}(v)$ 
and $\tau_{h_\ion}(v)$, model the nonlinear gating 
of the ionic conductance.
The activation functions are given by 
sigmoid functions
of the form
\begin{equation}
	\label{eq:sigmoid}
	\sigma(v) = \frac{1}{1+\exp\left(-(v-\vhalf)/
	\slope\right)},
\end{equation}
where the constants $\vhalf_{m_\ion}\in \setreal$ and $\vhalf_{h_\ion} \in \setreal$ determine the \textit{half-activation} of those functions, while the constants $\slope_{m_\ion} > 0$ and $\slope_{h_\ion} < 0$ determine their slopes. Because $\sigma_{m_\ion}:\setreal\to(0,1)$ and $\sigma_{h_\ion}:\setreal\to(0,1)$ are monotonically increasing and decreasing, respectively, the states  $m_\ion$ and $h_\ion$ are called activation and inactivation \textit{gating variables}, respectively. The time-constant functions are given by bell-shaped functions of the form\footnote{Some models are defined with different types of sigmoids and bell-shaped functions. The results in this paper can be trivially extended to those
cases.}
\begin{equation}
	\label{eq:gaussian_function}
	\tau(v) = \taumin + 
	(\taumax-\taumin) \exp(-(v-\mean)^2/\std^2) 
\end{equation}
for all $v \in \setreal$ and some $\taumin,\taumax > 0$ and $\mean,\std \in \setreal$. 

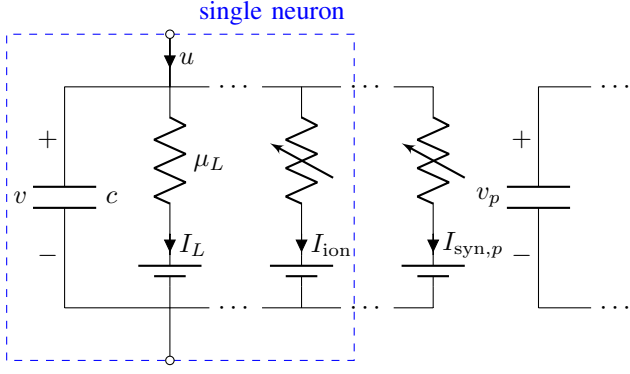
\begin{figure}[t]
	\centering
	\if\usetikz1
	\ctikzset{label/align = smart}
	\begin{tikzpicture}[scale=0.70,
		american resistors,american voltages]
		\newcommand{\cbtop}{4.2}

		\draw (0,0) to[C,l_=$c$,v^<=$v\;\;$] 
		(0,\cbtop);
		
		\draw[draw=blue,dashed] (-1.1,-1) rectangle (5.5,\cbtop+1) node[above left]{{\color{blue}single neuron}};
		\draw (0,\cbtop) to (2.75,\cbtop);
		\draw (0,0) to (2.75,0);
		\node (dots1) at (3.25,\cbtop)[] {$\cdots$}; 
		\node (dots2) at (3.25,0)[] {$\cdots$};
		\draw (3.75,\cbtop) to (5.25,\cbtop);
		\draw (3.75,0) to (5.25,0);
		\node (dots3) at (5.75,\cbtop)[] {$\cdots$}; 
		\node (dots4) at (5.75,0)[] {$\cdots$};
		\draw (6.25,\cbtop) to (7,\cbtop);
		\draw (6.25,0) to (7,0);
		
		\draw (2,\cbtop) to[R,l=$\maxcond_L$] 
		(2,\cbtop/3)to[battery1,i>=$I_L$] (2,0);

		\draw (4.5,\cbtop) to[vR] (4.5,\cbtop/3)
		to[battery1,i>=$\Iint_\ion$] (4.5,0);
			
		\draw (7,\cbtop) to[vR] (7,\cbtop/3)
		to[battery1,	i>=$\Isyn_{\syn,\pre}$] (7,0);
	
		\draw (2,\cbtop+1) 
			to[short,o-,i=$\Iapp_\app$] (2,\cbtop);
		\draw (2,0) to[short,-o] (2,-1);

		\draw (9,0) to[C,v^<=$v_\pre\;\;$] 
		(9,\cbtop);
		\draw (9,\cbtop) to (10,\cbtop);
		\draw (9,0) to (10,0);
		\node (dots3) at (10.5,\cbtop)[] {$\cdots$}; 
		\node (dots4) at (10.5,0)[] {$\cdots$};
	\end{tikzpicture}
	\else
		\includegraphics[scale=1]{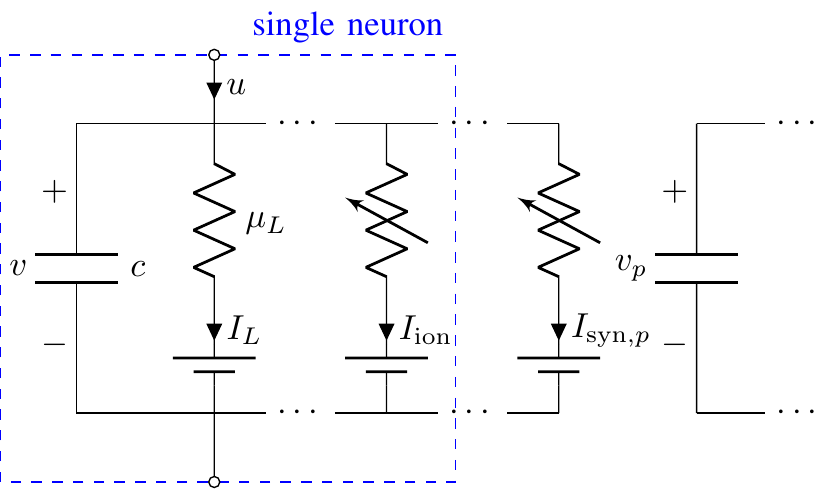}
	\fi
	\caption{Circuit representation of a neuron with voltage $v$ that is coupled	though a synapse to a presynaptic neuron with voltage $v_\pre$.}
	\label{fig:conductance_based}
\end{figure}

\begin{ex}
	\label{ex:HH_description} 
	The Hodgkin-Huxley (HH) model \cite{hodgkin_quantitative_1952} includes two intrinsic ionic currents: a transient sodium current $\INa$ and a potassium current $\IK$, so that $\mathcal{I} = \{\rm{Na},\rm{K}\}$. The voltage dynamics of the	HH model are given by 
	\[
	\begin{split}
		c \, \dot{v} = 	
				 	-\underbrace{\maxcond_{\rm{Na}} m_{\rm{Na}}^3 h_{\rm{Na}} 
				 	(v-\nernst_{\rm{Na}})}_{\INa} 
				 	&-\underbrace{\maxcond_{\rm{K}} m_{\rm{K}}^4 
				 	(v-\nernst_{\rm{K}})}_{\IK} \\
					&-\underbrace{\maxcond_{\rm{L}}(v-\nernst_{\rm{L}})}_{\IL}
					\,+\, \Iapp_\app,
	\end{split}
	\]	
	the dynamics of $m_\Na$ and $m_\K$ are given by \eqref{eq:activation}, and the dynamics of $h_\Na$ are given by \eqref{eq:inactivation}.
\eoe
\end{ex}

Two basic properties of a single neuron conductance-based model justify the assumptions of \cref{sec:model_structure}. The first property is the existence of a positively invariant set.
\begin{lem}
	\label{lem:invariant_set}
	Consider the neuronal model \eqref{eq:single_neuron_cb}-\eqref{eq:gaussian_function}, and assume $\mid u\mid  \le \umax$ for all $t \ge 0$.	
	 Let 
	\begin{equation}
	\label{eq:v_bounds} 
		\begin{split}
		\vmax &:= 
		\max\left\{\max_{\ion\in\mathcal{I}} \nernst_\ion,\;	
		\umax \,\maxcond_{\rm{L}}^{-1}
		+ \nernst_{\rm{L}}
		\right\} \\
		\vmin &:= 
		\min\left\{\min_{\ion\in\mathcal{I}}\nernstint_\ion, \;
		- \umax \,\maxcond_{\rm{L}}^{-1}		
		+ \nernst_{\rm{L}}
		\right\}
		\end{split}
	\end{equation}
	Whenever $v(0) \in [\vmin,\vmax]$, $m_\ion(0) \in [0,1]$ and $h_\ion(0) \in [0,1]$, it follows that
	\begin{equation*}
	v(t) \in [\vmin,\vmax],\quad
	m_\ion(t) \in [0,1], \text{ and }
	h_\ion(t) \in [0,1]
	\end{equation*}	
	for all $\ion \in \mathcal{I}$ and
	all $t \ge 0$.
\end{lem}
\begin{proof}
	See \cref{proof:invariant_set}.
\end{proof}

The second basic property of a conductance-based model is the contraction of its internal dynamics.
\begin{lem} 
	\label{lem:int_dyn_contraction}
For all $\ion \in \mathcal{I}$, the dynamics \eqref{eq:activation} are globally exponentially contracting, uniformly in $v$ on $\setreal$ and in $\{\vhalf_{m_\ion}, \slope_{m_\ion}, \mean_{m_\ion}, \std_{m_\ion}\}$ on 
$\setreal^4$. Exponential contraction holds for any (scalar) constant contraction metric and a contraction rate given by $2 \taumax_{m_{\ion}}^{-1}$. The same holds analogously for the dynamics \eqref{eq:inactivation}.
\end{lem}
\begin{proof}
	The Jacobian of the vector field of \eqref{eq:activation} is $-\tau_{m_{\ion}}^{-1}(v)$. But from \eqref{eq:gaussian_function}, we see that for any $p>0$, the inequality
	\[
		-\tau_{m_{\ion}}^{-1}(v)\,p
		-p\,\tau_{m_{\ion}}^{-1}(v)
		< -2\,\taumax_{m_{\ion}}^{-1}\,p 
	\]
	holds for any real $v$,
	$\vhalf_{m_\ion}$, $\slope_{m_\ion}$, $\mean_{m_\ion}$, and $\std_{m_\ion}$.
\end{proof}

Finally, we formalize the connection between the single neuron model above and the model structure \eqref{eq:true_system}.
\begin{prop}
\label{prop:parametrization} 
Consider the neuronal model \eqref{eq:single_neuron_cb}-\eqref{eq:gaussian_function}. Let
\[
	w := 
	\begin{pmatrix}	m_{\ion_1},h_{\ion_1},m_{\ion_2},h_{\ion_2},\dotsc
	\end{pmatrix}^\transp,
\]
and let the parameter vector $\parw$ be composed of any number of elements from the set 
\[\cup_{i=1}^\nw \{\vhalf_{w_i},\slope_{w_i},\mean_{w_i},\std_{w_i}\}.\]
Let $\parv$ be defined according to one of the following parametrizations:
\begin{align}
\nonumber
\parv &:= \begin{pmatrix}
			\maxcond_{\ion_1},\,\maxcond_{\ion_2},\,\dotsc
			\end{pmatrix}^\transp, \text{ or } \\ 
\nonumber			
\parv &:= c^{-1}\begin{pmatrix}
			1,\,\maxcond_{\ion_1},\,\maxcond_{\ion_2},\,\dotsc
			\end{pmatrix}^\transp, \text{ or } \\ 
\nonumber			
\parv &:= c^{-1}\begin{pmatrix}
			1,\,\maxcond_{\ion_1},\,\maxcond_{\ion_2},\,\dotsc,\,
			\maxcond_{\ion_1}\nernst_{\ion_{1}},\,\maxcond_{\ion_2}\nernst_{\ion_{2}},\,\dotsc
			\end{pmatrix}^\transp.
\end{align}
Then the neuronal model is of the form \eqref{eq:true_system}, and it satisfies \cref{assum:true_invariant_set,assum:int_dyn_contraction}.
\end{prop}
\begin{proof}
	 Given the above parametrization, it can be verified by inspection that \eqref{eq:single_neuron_cb}-\eqref{eq:gaussian_function} can be written as \eqref{eq:true_system}. It then follows immediately from \cref{lem:invariant_set,lem:int_dyn_contraction} that the neuronal model satisfies \cref{assum:true_invariant_set,assum:int_dyn_contraction}.
\end{proof}

As \cref{prop:parametrization} points out, a conductance-based model can be parametrized in a number of ways. The choice of parametrization depends on implicit assumptions about which model constants are known, and which need to be estimated. Estimation of  the maximal conductances $\maxcond_\ion$ is of particular importance in neurophysiological applications, as they can be regarded as the key parameters for {\it adaptive} control of a neuronal network
\cite{drion_neuronal_2015,drion_cellular_2019}. 
Maximal conductances 
vary greatly under the biochemical action of neuromodulators \cite{marder_neuromodulation_2014}.
In contrast, many other constants in a neuron model may be assumed to be known, but with some level of uncertainty.
\begin{ex}
	\label{ex:HH_parametrization} 
	Assume that the capacitance, maximal conductances and  half-activations of the HH model of \cref{ex:HH_description} need to be estimated, while other model constants are known. Then we may define $w:=(m_\Na,h_\Na,m_\K)^\transp$ and
\[
\begin{split}
	\parv &:= 
	c^{-1}
	\begin{pmatrix}
	1,\,
	\mu_\Na,\,\mu_\K,\,
	\mu_\Leak
	\end{pmatrix}^\transp \\
	\parw &:= 
	\begin{pmatrix}
	\vhalf_{m_\Na},\,\vhalf_{h_\Na},\,
	\vhalf_{m_\K}
	\end{pmatrix}^\transp
\end{split}
\]
so that the dynamics of the HH model are given by \eqref{eq:true_system}, with
\[
	\begin{split}
	\Phi(v,w,u) &= 
	-\begin{pmatrix}
	-u, w_1^3w_2(v-\nernst_\Na),
	w_3^4(v-\nernst_\K),(v-\nernst_\Leak)
\end{pmatrix} \\[.3em]
	A(v) &= -\diag\left(
			\tau_{m_\Na}^{-1}(v),
			\tau_{h_\Na}^{-1}(v),
			\tau_{m_\K}^{-1}(v)
			\right)\\[.3em]
	\bw(v,\parw) &= -A(v)
	\col\left(
	\sigma_{m_\Na}(v),
		\sigma_{h_\Na}(v),		
			\sigma_{m_\K}(v)
	\right)
	\end{split}
\]
and $\bv = 0$.
\eoe
\end{ex}

\subsection{Conductance-based neural network model}
\label{sec:cb_networks} 

The two basic properties of single neuron models discussed in the previous section extend to conductance-based network models. 
A conductance-based neural network is given by the interconnection of $\nv \in \setnat$ single neurons via \textit{synapses}. For $ i \in \mathcal{N} := \{1,\dotsc,\nv\}$, the voltage dynamics of the $i^{\rm{th}}$ neuron in the network is described by
\begin{equation}
	\label{eq:neuron_in_network}
		c \, \dot{v}_i = 
		- I_{\rm{L},i} 
		- \sum_{\ion \in \mathcal{I}} 
		\Iint_{\ion,i}
		-\sum_{\syn \in \mathcal{S}} 
		\sum_{\pre \in \mathcal{P}} 
		\Isyn_{\syn,\pre,i}
		+
		u_i
\end{equation}
where each $I_{\rm{L},i}$ is given by \eqref{eq:leak_current} and each $\Iint_{\ion,i}$ is given by \eqref{eq:current_cb}, as before (in this case a subscript $i$ is attached to all variables). The additional currents $\Isyn_{\syn,\pre,i}$ above are \textit{synaptic currents} interconnecting the $i^{\text{th}}$ (postsynaptic) neuron with the $p^{\text{th}}$ (presynaptic) neuron, so that $\mathcal{P} \subseteq \mathcal{N}$. Since there might exist multiple synapses (based on different neurotransmitters) connecting two neurons, we denote each synaptic type by $\syn$, and the index set of synaptic types by $\mathcal{S}$.

Synaptic currents arise from electrochemical connections between neurons
\cite[Chapter 7]{ermentrout_mathematical_2010}. We consider
the model used in \cite{ermentrout_mathematical_2010,dethier_positive_2015},
which can be written as
\begin{subequations}
	\label{eq:syn_dynamics}
	\begin{align}
		\label{eq:syn_current} 
		\Isyn_{\syn,\pre} &= \maxcond_{\syn,\pre} s_{\syn,\pre} \,
		(v - \nernst_\syn)
				\\[.5em]
		\label{eq:syn_activation} 
		\tau_\syn(v_\pre)
		\dot{s}_{\syn,\pre} &=  
		-s_{\syn,\pre}+
		a_\syn \tau_\syn(v_\pre) 
		\sigma_\syn(v_\pre)
	\end{align}
\end{subequations}
with a synaptic time-constant function $\tau_\syn$ given by
\begin{equation}
	\label{eq:synaptic_time-constant} 
	\tau_\syn(v_\pre) = \frac{1}{
	a_\syn 
	\sigma_\syn(v_\pre)	
	+b_\syn}
\end{equation}
and a synaptic activation function 
$\sigma_{\text{syn}}$ of the form
\eqref{eq:sigmoid}, with $\vhalf_{\text{syn}} \in \setreal$ and 
$\slope_{\text{syn}} > 0$. Here, $s_{\syn,\pre}$ is the synaptic gating
variable, $v_\pre$ is the membrane voltage of the presynaptic neuron, and 
$a_\syn>0$ and $b_\syn>0$ are constant parameters. 
The constants $\maxcond_{\syn,\pre} > 0$ and 
$\nernst_\syn\in\setreal$ are (synaptic) maximal conductances and reversal
potentials, respectively. Notice that $0<(a_\syn+b_\syn)^{-1} \le \tau_\syn(v_\pre) \le b_\syn^{-1}$ for all $v_\pre \in \setreal$.

\begin{prop}
\label{prop:network_parametrization} 
Consider a conductance-based neural network model with voltage output vector  
$v = (v_1,\dotsc,v_\nv)^\transp$
and internal state vector 
$w := \col(w^{(1)},\dotsc,w^{(\nv)}) $
where 
\[ 
	w^{(i)} := \begin{pmatrix}
	m_{\ion_1,i},\,
	\dotsc,\,s_{\syn_1,1,i},\,
	s_{\syn_2,1,i}\,\dotsc,\,
	s_{\syn_1,2,i},\,\dotsc
	\end{pmatrix}
\] collects the intrinsic and synaptic gating variables of the $i^{\rm{th}}$ neuron. Let each neuron in the network be parametrized as in \cref{prop:parametrization}, allowing for the inclusion of $\maxcond_\syn$ and $\nernst_\syn$ in $\parv$, and for the inclusion of $\vhalf_\syn$, $\slope_\syn$, $\mean_\syn$ and $\std_\syn$ in $\parw$.
Then the network model is of the form \eqref{eq:true_system}, and it satisfies \cref{assum:true_invariant_set,assum:int_dyn_contraction}.
\end{prop}

Since  \cref{prop:network_parametrization} is a trivial extension of \cref{prop:parametrization}, we omit its proof and present a concrete example instead:
\begin{ex}
	\label{ex:HCO_description} 
	A Half-Center Oscillator (HCO) is a circuit composed of two neurons mutually coupled by inhibitory synapses. This elementary network is the simplest example of a Central Pattern Generator, a type of neural network that plays an important role in the generation of autonomous rhythms for motor control \cite{marder_central_2001}.
	A simple HCO model is obtained 	by interconnecting two HH neurons with a GABA-type\footnote{Gamma-aminobutyric acid (GABA) is a neurotransmitter associated with inhibitory synapses.} synaptic current $I_\GABA$, and adding to each of the neurons an intrinsic calcium current $I_\Ca$ \cite{dethier_positive_2015}. This results in $\mathcal{I} = \{\Na,\K,\Ca\}$, 
	$\mathcal{S} = \{\GABA\}$, and voltage dynamics given by
	\begin{equation*}
	\begin{split}
	c_i \dot{v}_i = 
	&-\maxcond_{\Na,i} m_{\Na,i}^3 h_{\Na,i}(v_i-\nernst_{\Na})
	-\maxcond_{\K,i} m_{\K,i}^4 (v_i-\nernst_{\K}) \\
	&-\maxcond_{\Ca,i}m_{\Ca,i}^3 h_{\Ca,i}(v_i-\nernst_{\Ca})
	-\maxcond_{\GABA,p,i}s_{\GABA,p,i}(v_i-\nernst_{\GABA}) \\
	&-\maxcond_{\Leak,i}(v_i-\nernst_{\Leak}) + \Iapp_i
	\end{split}
	\end{equation*}
	for $i,p\in \mathcal{N}=\{1,2\}$ and $p\neq i$. The gating
	variables of each neuron, which evolve according to 
	\eqref{eq:activation}-\eqref{eq:activation} and
	\eqref{eq:syn_activation}, are collected in 
	$w^{(i)} = (m_{\Na,i},h_{\Na,i},m_{\K,i},m_{\Ca,i},h_{\Ca,i},s_{\GABA,i,p})^ \transp$.
Now let 
\begin{equation}
	\label{eq:maxcond_parametrization_1} 
	\maxcond^{(i)} = (\maxcond_{\Na,i},\maxcond_{\K,i},\maxcond_{\Ca,i},
	\maxcond_{\GABA,p,i},\maxcond_{\Leak,i})^\transp
\end{equation}
for $i,p\in \mathcal{N}=\{1,2\}$ and $p\neq i$. Then we can parameterize
the HCO according to
\begin{equation}
	\label{eq:maxcond_parametrization_2} 
	\parv = \col(\maxcond^{(1)},\maxcond^{(2)})
\end{equation}
with $\maxcond^{(1)}$ and $\maxcond^{(2)}$ given by \eqref{eq:maxcond_parametrization_1}.
Letting $v=(v_1,v_2)^\transp$ and $w = \col(w^{(1)},w^{(2)})$, the voltage dynamics of the model can then be written as \eqref{eq:dv_true}, where
\[		
	\begin{split}
		\Phi(v,w) &=
		\begin{bmatrix}
			\varphi(v_1,w^{(1)}) & 0 \\
			0 & \varphi(v_2,w^{(2)})
		\end{bmatrix} \\
		\bv(t) &= (\Iapp_1(t)/c_1, \Iapp_2(t)/c_2)^\transp
	\end{split}
\]
with
\[
		\varphi(v_i,w^{(i)}) = 
		-\frac{1}{c_i}
		\begin{pmatrix}
		m_{\Na,i}^3 h_{\Na,i}(v_i-\nernst_{\Na})
		\\[.5em]
		 m_{\K,i}^4 (v_i-\nernst_{\K})
		\\[.5em]
 		m_{\Ca,i}^3 h_{\Ca,i}(v_i-\nernst_{\Ca})
		\\[.5em]
		s_{\GABA,p,i}(v_i-\nernst_\GABA)
		\\[.5em]
		v_i-\nernst_\Leak
		\end{pmatrix}^\transp
\]
for $i=1,2$ and $p\neq i$.
\eoe
\end{ex}

\section{Estimation of the output dynamics}
\label{sec:no_uncertainty} 

In this section, we simplify the problem statement of \cref{sec:model_structure} by considering the case in which there is no uncertain parameter $\eta$, that is, the internal dynamics are assumed to be perfectly known. Hence we consider the simplified model 
\begin{subequations}
\label{eq:true_system_simple} 
	\begin{align}
		\label{eq:dv_true_simple} 
		\dot{v} &= 
		\Phi(v,w,u) \parv 
			+ \bv(v,w,u) \\
		\label{eq:dw_true_simple} 
		\dot{w} &= 
		A(v) w + \bw(v) \\
		\label{eq:dtheta_true_simple}
	  	\dot{\parv} &= 0
	\end{align}
\end{subequations}
satisfying \cref{assum:true_invariant_set,assum:int_dyn_contraction}. This simplified model already deserves attention. 
In \cref{sec:challenges}, we illustrate how a naive output-error estimation scheme leads to issues. Then, 
Using the basic properties of
\cref{sec:modeling}, we propose a  least squares method in \cref{sec:batch_forgetting}, and a recursive least squares (RLS)--based adaptive observer in \cref{sec:simple_observer}.
 
\subsection{Challenges in neuronal model estimation}
\label{sec:challenges}

The neuronal behaviours displayed by conductance-based models range from the simple spiking oscillations of single neurons to the large-scale rhythmic computations performed by cortical networks \cite{yuste_cortex_2005}. This is because despite the fact that conductance-based models have contracting internal dynamics, their overall dynamics are in general non-contracting. In the terminology of linear systems, they have stable zeros, but possibly unstable poles. In biophysical terms, this happens due to intrinsic ionic currents with a negative differential conductance, a source of positive feedback and instability 
\cite{sepulchre_control_2019-1}.

\begin{ex}
	\label{ex:positive_feedback} 
	Consider the HH model of	\cref{ex:HH_description}, with
	reversal potentials such that 
	$\nernst_\K< \nernst_{\rm{L}} < \nernst_\Na$. For any input 
	current such that $|\Iapp_\app(t)| <	\maxcond_{\rm{L}} (\nernst_\Na -
	 \nernst_{\rm{L}})$ for all $t\ge 0$, the voltage bounds from \cref{lem:invariant_set} show that the membrane voltage
	satisfies $v < v_\Na$ for all $t \ge 0$. This implies that 
	\[
		-\partial_{m_\Na} I_\Na \, \partial_{v} \dot{m}_\Na > 0
	\]
	at any admissible equilibrium of the system,
	and thus the Sodium current $\INa$ introduces positive feedback
	to the membrane voltage of the HH model.
	\eoe
\end{ex}

The non-contracting nature of neuronal dynamics is the main reason why traditional parameter estimation methods based on \textit{output-error} (or \textit{simulation-error}) criteria \cite{ljung_system_1999,manchester_identification_2011-1} cannot be effectively applied to conductance-based models. This is illustrated by means of a numerical example:
\begin{ex} 
	\label{ex:output_error} 
	Consider the typical biophysical parameters of the HH model shown in
	\cref{tab:HH_nominal_bio_parameters} below. Owing to the large value of 
	$\maxcond_\Na$, the positive feedback introduced by the Sodium current
	$I_\Na$ dominates the model dynamics in some regions of the state-space
	\cite{izhikevich_dynamical_2007}. Using these parameters, Figure
	\ref{fig:HH_spike} illustrates the excitability of the model.
	
	\begin{table}[b]
		\caption{Parameters of the HH model
		\cite{hodgkin_measurement_1952,izhikevich_dynamical_2007}.}
		\normalsize
		\centering
		\begin{tabular}{|p{.8cm}|p{.8cm}|p{.8cm}|p{.8cm}|
			p{.8cm}|p{1.0cm}|p{.6cm}|}
		\hline
		$\maxcond_\Na$ &	$\maxcond_\K$ & $\maxcond_{\rm{L}}$ &
		$\nernst_\Na$ & $\nernst_\K$ & $\nernst_{\rm{L}}$ &
		$c$ 
		\\
		\hline
		$120$ & $36$ & $0.3$ & $55$ & $-77$ & $-54.4$ & $1$ 
		\\
		\hline
		\end{tabular}
		\label{tab:HH_nominal_bio_parameters} 
	\end{table}

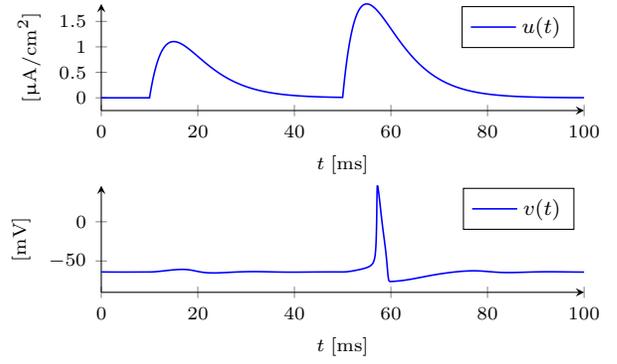
\begin{figure}[t]
	\centering
	\if\usetikz1
	\begin{tikzpicture}
		\begin{groupplot}[
			group style={group size=1 by 2,
			},
			height=3.0cm,width=8.0cm, 
			filter discard warning=false,
			axis y line = left,
			xlabel={$t \; \mathrm{[ms]}$},
			axis x line = bottom,	
			tick label style={font=\scriptsize},	
			label style={font=\scriptsize},
			legend style={font=\footnotesize},
			]
			\nextgroupplot[ymin=-0.25,ylabel=$\mathrm{[\upmu A / cm^2]}$]
			\addplot[color=blue,semithick]
				table[x index=0,y index=1] 
				{./data/HH_spike.txt};
				\addlegendentry{$\Iapp(t)$}; 
			\nextgroupplot[ymin=-90,ylabel=$\mathrm{[mV]}$]
			\addplot[color=blue,semithick]
				table[x index=0,y index=2] 
				{./data/HH_spike.txt};
				\addlegendentry{$v(t)$}; 
		\end{groupplot}
	\end{tikzpicture}
	\else
		\includegraphics[scale=1]{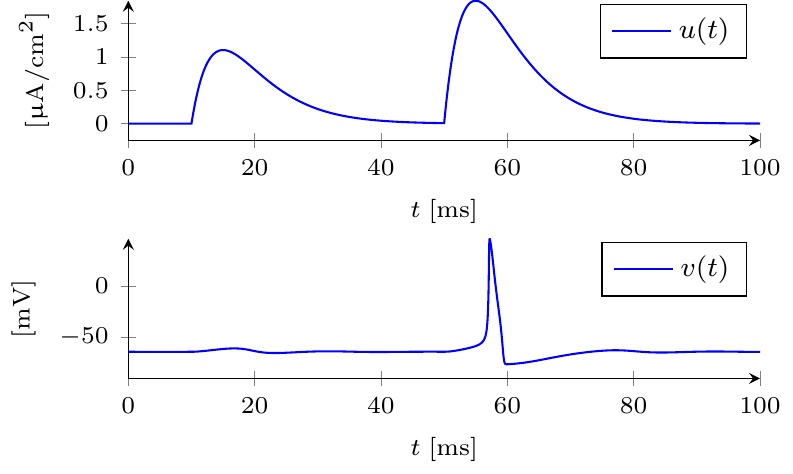}
	\fi
	\caption{Excitability in the HH model. A small current pulse causes no spike, 
	while a larger current pulse causes a spike. 
	HH parameter values are described in \cref{tab:HH_nominal_bio_parameters}
	and \cref{sec:HH_parameters}.}
	\label{fig:HH_spike}
\end{figure}
	
	Now, consider the parametrization given by 
	$\parv = \maxcond_\Na$, and suppose $\maxcond_\Na$ must
	be estimated from the continuous-time measurements 
	$\Iapp(t)$ and $v(t)$ shown in Figure
	\ref{fig:HH_spike}. In a naive application of the
	prediction-error method \cite{ljung_system_1999},
	a \textit{predictor model} would be given by
	\begin{equation}
	\label{eq:estimator_batch} 
	\begin{split}
		c \dot{\vest} &= -\maxcondest_\Na 
		\mest_\Na^3 \hest_\Na (\vest - \nernst_\Na) + 
		\bv(\vest,\mest_\K,\Iapp)
		\\
		\dot{\west} &= 
		A(\vest) \west + \bw(\vest)
	\end{split}
	\end{equation}
	with $\west = (\mest_\Na,\hest_\Na,\mest_\K)^\transp$ and
	\[
		\bv(\vest,\mest_\K,\Iapp) = 
		-\maxcond_\K \mest_\K^4(\vest-\nernst_\K) 
		- \maxcond_\Leak(\vest-\nernst_\Leak) + \Iapp 
	\]
	The estimate $\maxcondest_\Na$ is then obtained by minimizing
	the \textit{output-error} 
	cost function
	\begin{equation}
		\label{eq:HH_cost} 
		V(\maxcondest_\Na,\west(0),T) 
		= \frac{1}{T} \int_0^T 
		(v(t) - \vest(t))^2 dt
	\end{equation}
	in $\maxcondest_\Na$ and $\west(0)$. The issue with
	this approach is that it may not be trivial to find a
	global minimum for $\maxcondest_\Na$ using numerical
	methods, even when the problem is simplified by fixing 
	$(\vest(0),\west(0)) = (v(0),w(0))$. 
	The reason can be visualized in Figure \ref{fig:HH_cost}, where we have
	plotted the cost function $V(\maxcondest_\Na,w(0),T)$ obtained with the
	input-output traces from Figure \ref{fig:HH_spike} ($T=100$).
	It can be seen that the cost function is nearly discontinuous between 
	$\maxcondest_\Na = 107$ and $\maxcondest_\Na = 108$.
	The behaviour underlying this near discontinuity is unveiled in Figure
	\ref{fig:HH_missedspike}, where two solutions of \eqref{eq:estimator_batch}
	are plotted corresponding to the estimates $\maxcondest_\Na = 107$ and 
	$\maxcondest_\Na = 108$. It is the spike which appears when increasing 
	$\maxcondest_\Na$ that causes a sudden change in the cost function. 
	The cost function contains one near discontinuity, since there is a single
	spike being fired. For a dataset with multiple spikes, the cost function
	would rapidly become intractable.
	
\begin{figure}[t]
	\centering
	\if\usetikz1
	\begin{tikzpicture}
		\begin{groupplot}[
			group style={group size=1 by 2,
			},
			height=4.0cm,width=8.0cm, 
			filter discard warning=false,
			axis y line = left,
			xlabel={$\maxcondest_\Na$}, 
			axis x line = bottom,	
			tick label style={font=\scriptsize},	
			label style={font=\scriptsize},
			legend style={font=\footnotesize},
			]
			\nextgroupplot[ymax=230,ymin=0]
			\addplot[color=blue,semithick]
				table[x index=0,y index=1] 
				{./data/HH_cost.txt};
				\addlegendentry{$V(\maxcondest_\Na,T)$}; 
			\draw [dashed] (108,0) -- (108,230);
			\draw [dashed] (107,0) -- (107,110);
			\node at (109.5,50)[] {\scriptsize $108$}; 
			\node at (105.5,50)[] {\scriptsize $107$};			
			\nextgroupplot[
				ymin = -50,
				ymax = 50,
				]
			\addplot[color=blue,semithick]
				table[x index=0,y index=1] 
				{./data/HH_dcost.txt};
				\addlegendentry{$\partial_{\maxcondest_\Na}V(\maxcondest_\Na,T)$}; 
		\end{groupplot}
	\end{tikzpicture}
	\else
		\includegraphics[scale=1]{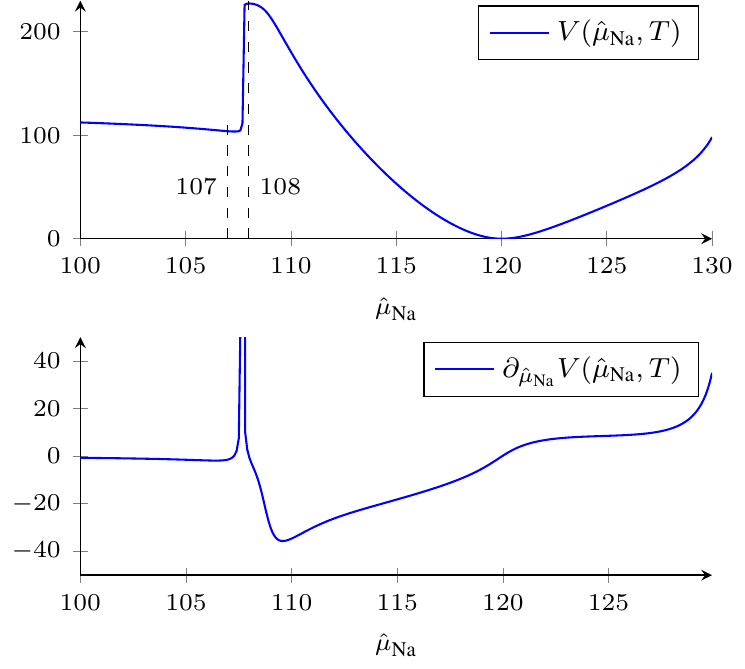}
	\fi
	\caption{Cost function $V(\maxcondest_\Na,w(0),T)$ given by 
	\eqref{eq:HH_cost} and its gradient. }
	\label{fig:HH_cost}
\end{figure}	
\eoe
\end{ex}

\cref{ex:output_error} illustrates the more general
problem of lack of tractability in estimating the
parameters of a non-contracting system with an 
output-error criterion
\cite{ribeiro_smoothness_2020,abarbanel_estimation_2008}.
In fact, the lack of contraction is
the root cause of what has been called the 
``exploding gradient'' problem in deep learning theory
\cite{pascanu_difficulty_2013}. The exploding gradient
is clearly visible in Figure \ref{fig:HH_cost}.

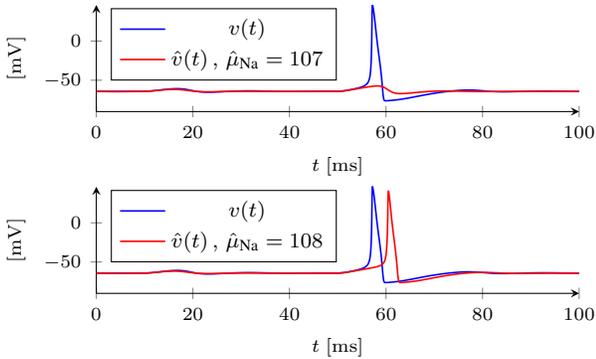
\begin{figure}[t]
	\centering
	\if\usetikz1
	\begin{tikzpicture}
		\begin{groupplot}[
			group style={group size=1 by 2,
			},
			height=3.0cm,width=8.0cm, 
			filter discard warning=false,
			axis y line = left,
			xlabel={$t \; \mathrm{[ms]}$},
			axis x line = bottom,	
			tick label style={font=\scriptsize},	
			label style={font=\scriptsize},
			legend style={font=\footnotesize},
			]
			\nextgroupplot[ymin=-90,ylabel=$\mathrm{[mV]}$,legend pos = north west]
			\addplot[color=blue,semithick]
				table[x index=0,y index=2] 
				{./data/HH_spike.txt};
				\addlegendentry{$v(t)$}; 
			\addplot[color=red,semithick]
				table[x index=0,y index=3] 
				{./data/HH_spike.txt};
				\addlegendentry{$\vest(t)\, , \, \maxcondest_\Na = 107$}; 
			\nextgroupplot[ymin=-90,ylabel=$\mathrm{[mV]}$,legend pos = north west]
			\addplot[color=blue,semithick]
				table[x index=0,y index=2] 
				{./data/HH_spike.txt};
				\addlegendentry{$v(t)$}; 
			\addplot[color=red,semithick]
				table[x index=0,y index=4] 
				{./data/HH_spike.txt};
				\addlegendentry{$\vest(t)\, , \, \maxcondest_\Na = 108$}; 
		\end{groupplot}
	\end{tikzpicture}
	\else
		\includegraphics[scale=1]{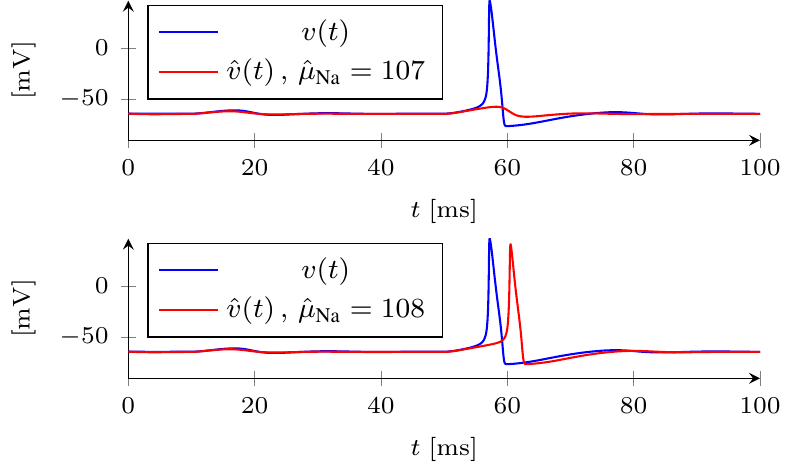}
	\fi
	\caption{Solutions of the HH predictor \eqref{eq:estimator_batch}
	for	two values of $\maxcondest_\Na$ (red), compared to the solution of the
	true model displayed in Figure \ref{fig:HH_spike} (blue).}
	\label{fig:HH_missedspike}
\end{figure}

\subsection{Least squares and reduced-order observer}
\label{sec:batch_forgetting} 

A simple least squares solution to the problem of estimating the parameters of \eqref{eq:true_system_simple} exploits the following observation:
\begin{rmk}
\label{rmk:reduced_observer} 
\cref{assum:int_dyn_contraction} implies that the system 
\begin{equation}
	\label{eq:reduced_observer} 
	\dot{\west} = A(v) \west + \bw(v)
\end{equation}
is a globally exponentially convergent reduced-order identity observer for the dynamics \eqref{eq:dw_true_simple}. More precisely, as $t~\to~\infty$ we have $\west(t) \to w(t)$ for any piecewise continuous $v(t)$ and any initial conditions $\west(0),w(0) \in \setreal^\nw$.
\end{rmk}

We employ the reduced-order observer \eqref{eq:reduced_observer} to
obtain estimates $\west$ of the internal states 
$w$. Contraction of the internal states
suggests postulating the predictor model
\begin{subequations}
	\label{eq:eq_error_estimator} 
\begin{align}
	\label{eq:predictor_voltage} 
	\dot{\vest} &= \Phi(v,\west,u)\hat{\parv} 
	+ \bv(v,\west,u) \\
	\label{eq:predictor_int} 
	\dot{\west} &= 
	A(v)\west + \bw(v)
\end{align}
\end{subequations}
which, 
for $\west(0)=w(0)$, 
reduces to a continuous-time \textit{equation-error} model structure
\cite{ljung_system_1999,astrom_adaptive_2008}. Classical system identification theory 
\cite[Section 2]{astrom_adaptive_2008} thus suggests performing parameter estimation by solving the regularized problem
\begin{equation}
	\label{eq:batch_problem}
	\hat{\parv}(T) = \min_{\hat{\parv}} V(\hat{\parv},T) + \hat{\parv}^\transp R_0(T) \, \hat{\parv}
\end{equation}
with the weighted cost function
\begin{equation}
	\label{eq:cost_forgetting} 
	V(\hat{\parv},T) = \frac{1}{T} \int_0^T 
	e^{-\alpha(T-\tau)}\|H\dot{v}(\tau)-H\dot{\vest}(\tau)\|^2 d\tau
\end{equation}
where $\alpha>0$ is a \textit{forgetting factor}, introduced to
discount the initial error between $w(0)$ and $\west(0)$, $R_0$ is a symmetric positive semidefinite matrix, and $H$ is the operator of a strictly proper LTI filter introduced to avoid differentiating $v(t)$. Choosing the simple filter
\begin{equation}
	\label{eq:filter} 
	H(s) = \frac{\gain}{s+\gain}
\end{equation}
leads to
\begin{equation}
	\label{eq:filtered} 
	\begin{split}
	H\dot{\vest}(t) &= \Psi(t) \hat{\parv} + H \hat{\bv}(t)
	\\
	\dot{\Psi}(t) &= -\gamma\Psi(t) + \gain \Phi(v(t),\west(t),u(t))\\
	\hat{\bv}(t) &= \bv(v(t),\west(t),u(t))
	\end{split}
\end{equation}
which shows that \eqref{eq:batch_problem}-\eqref{eq:cost_forgetting} is now quadratic in $\hat{\parv}$ (notice $H\dot{v}$ is obtained by filtering the data $v$ with $sH(s)$). It follows that the batch problem \eqref{eq:batch_problem} admits a well-known solution based on the normal equation \cite[p. 55]{astrom_adaptive_2008}. That $\hat{\parv}(T) \to \parv$ as $T \to \infty$ will be shown to be a consequence of the convergence properties of the adaptive observer introduced in \cref{sec:simple_observer}.

\subsection{RLS-based adaptive observer}
\label{sec:simple_observer} 

Consider the system \eqref{eq:true_system_simple}. An adaptive observer for this system is given by
\begin{equation}
	\label{eq:adaptive_observer}
	\begin{split}
	  	\dot{\yest} &= 
		\Phi(v,\west,u)\hat{\parv} 
		+ \bv(v,\west,u)
	  	+ (\gain I +\Psi P \Psi^\transp) 
	  	(v-\yest)
	  	\\
	  	\dot{\west} &= 
		A(v)\west + \bw(v)
	  	\\
	  	\dot{\hat{\parv}} &= \gain P \, 
	  	\Psi^\transp \, (v-\yest)
	\end{split}
\end{equation}
where $\gain>0$ is a constant gain, and the matrices $P$ and $\Psi$  evolve according to
\begin{subequations}
	\label{eq:adaptive_matrices} 
	\begin{align}
	  	\label{eq:dPsi}
	  	\dot{\Psi} &= 
		- \gain \Psi + 
		\gain \Phi(v,\west,u), & \Psi(0) = 0
	  	\\
	  	\label{eq:dP} 
	  	\dot{P} &= \alpha P -  
	  	P \, \Psi^\transp \Psi P, & P(0) \succ 0
	\end{align}
\end{subequations}
where $\alpha>0$ is a constant forgetting factor.
The assumption that $\Psi(0) = 0$ is made without loss of generality.

\begin{rmk}
The adaptive observer \eqref{eq:adaptive_observer}-\eqref{eq:adaptive_matrices} relates to a number of designs in the literature. 
For instance, when we remove the internal dynamics ($\nw = 0$) and set $\alpha=\gamma$, then \eqref{eq:adaptive_observer}-\eqref{eq:adaptive_matrices} is similar to the high-gain design proposed in \cite{farza_adaptive_2009}.
Also, if $w(t)$ is assumed to be known, then by replacing $\hat{w}$ by $w$ in \eqref{eq:adaptive_observer} we recover a nonlinear variant of the classical linear design of \cite{zhang_adaptive_2001}. Finally, setting $P=I$ and $\Psi = \Phi$, and removing the adaptive gain and its dynamics \eqref{eq:adaptive_matrices}, it reduces to the design proposed in \cite{besancon_remarks_2000}, which can be thought of as being based on the Least-mean squares algorithm rather than RLS.
\end{rmk}
We will show that the convergence of the adaptive observer above does not require a high gain; a discussion on the benefits of tuning $\alpha$ and $\gamma$ will also be presented in \cref{sec:robustness}. Furthermore, we can prove the design \eqref{eq:adaptive_observer}-\eqref{eq:adaptive_matrices} is directly connected to the least squares problem discussed earlier:
\begin{prop}
	\label{prop:rls} 
	The adaptive observer \eqref{eq:adaptive_observer}-\eqref{eq:adaptive_matrices} implements the recursive least squares (RLS) solution of the least squares problem  \eqref{eq:batch_problem}-\eqref{eq:filtered}, with $R_0(T) = e^{-\alpha T} P^{-1}(0)/T$.
\end{prop}
\begin{proof}
	See \cref{proof:rls}.
\end{proof}

To show exponential convergence of the adaptive observer, we require a standard persistent excitation condition (see, for instance, \cite{besancon_remarks_2000,tyukin_adaptation_2007,farza_adaptive_2009}):
\begin{defn}
	\label{def:PE} 
	A time-varying matrix $M(t)$ is said to be persistently exciting (PE) if there exist $T>0$ and $\delta > 0$ such that for all $t \ge 0$, we have 
	\begin{equation*}
		\int_{t}^{t+T}
		M(\tau) M(\tau)^\transp d\tau \succeq \delta I	
	\end{equation*}
\end{defn}

\begin{assum}
	\label{assum:persistent_excitation}
	The signals $v(t)$ and $u(t)$ are such that for any trajectory of \eqref{eq:adaptive_observer}, the matrix $\Psi(t)^\transp$ is persistently exciting.
\end{assum}

It is well-known \cite{zhang_adaptive_2001} that \cref{assum:persistent_excitation} ensures uniform positive-definiteness of $P(t)$. In our context, we have:
\begin{lem}
	\label{lem:P}  
	Under \cref{assum:true_invariant_set,assum:int_dyn_contraction,assum:persistent_excitation}, the solution $\Psi(t)$ of \eqref{eq:dPsi} is bounded for all $t \ge 0$, and $P(t)$ is bounded and uniformly positive definite for all $t \ge 0$. In particular,
	\begin{equation}
	\label{eq:P_bounds}
		0 \preceq \Pmin I \prec P(t) \preceq \Pmax I
	\end{equation}
	for all $t \ge T$, with 
	\begin{equation}
		\label{eq:kappa1} 
		\begin{split}
			\Pmin &= 
\big(\|P^{-1}(0)\|+\alpha^{-1}\,\overline{\phi}^2\big)^{-1}
			 \\
			\Pmax &= 
			\delta^{-1} e^{2 \alpha T}
		\end{split}
	\end{equation}
	with $\overline{\phi} = \sup_{v\in V,w\in W,u\in U} \|\Phi(v,w,u)\|$.
\end{lem}
\begin{proof}
	See \cref{proof:P}.
\end{proof}

We can now state a global convergence result\footnote{\cref{thm:convergence_linear} can also be proven with classical (non-differential) Lyapunov arguments, in the fashion of \cite{zhang_adaptive_2001,farza_adaptive_2009}. Our proof relies instead on contraction (differential) analysis, which is useful to the results of the next section.} for the simple adaptive observer \eqref{eq:adaptive_observer}-\eqref{eq:adaptive_matrices}. 

\begin{thm}
	\label{thm:convergence_linear} 
	Consider the systems \eqref{eq:true_system_simple} and	\eqref{eq:adaptive_observer}-\eqref{eq:adaptive_matrices}, and let	    \cref{assum:true_invariant_set,assum:int_dyn_contraction,assum:persistent_excitation} hold. Let $\gamma > 0$ and $\alpha > 0$. Then, globally, we have	
	\[
		\col(\vest(t),\west(t),\hat{\parv}(t)) \to \col(v(t),w(t),\parv)
	\]
	exponentially fast as $t \to \infty$, with a convergence rate given by arbitrary $\ratevirt < \min\{\alpha,\rateint,\gamma\}$.
\end{thm}
\begin{proof}
	See \cref{proof:convergence_linear}.
\end{proof}

One should notice that the persistent excitation  \cref{assum:persistent_excitation} is classical yet difficult to check in practice, since it depends on system trajectories of a nonlinear system. However, the excitable (spiking) behavior of the neuronal circuits considered in this paper is an excellent source of excitation that can be reliably tapped through the application of superthreshold applied currents.

\section{Estimation under uncertainty}
\label{sec:main_result} 

The design in the previous section assumes no uncertainty in the model, which is unrealistic in a biophysical context. This section addresses different forms of uncertainty. \cref{sec:advanced_observer} deals with \textit{structured uncertainty} in the internal dynamics, modelled by the uncertain parameter $\parw$ in \eqref{eq:dw_true}. Building on the design \eqref{eq:adaptive_observer}-\eqref{eq:adaptive_matrices}, we present a locally convergent adaptive observer capable of estimating $\parw$ in addition to the unknown parameters $\parv$ in \eqref{eq:dv_true}. \cref{sec:noise} then discusses the problem of measurement errors and how the adaptive observer can be modified to mitigate that problem. Finally, \cref{sec:robustness} discusses the robustness of the adaptive observers with respect to unstructured uncertainty.

\subsection{Estimating uncertain internal dynamics parameters}
\label{sec:advanced_observer}

To deal with structured uncertainty in the internal dynamics of \eqref{eq:true_system}, we augment the simple adaptive observer \eqref{eq:adaptive_observer} to estimate the parameter vector $\parw$ as well. To design the observer the following is assumed: 

\begin{assum}
	\label{assum:known_sets} 
	Assume that compact sets $\Parv\in\setreal^{\nparv}$ and $\Parw \in \setreal^{\nparw}$ are known such that $\parv\in\Parv$ and $\parw\in\Parw$.
\end{assum}

\begin{rmk}
\label{rmk:saturation_parameters} 
Analogously to \cref{rmk:saturation_w}, under \cref{assum:true_invariant_set,assum:known_sets}, we can assume without loss of generality that for all $v \in V$, the functions $A(v,\parw)$ and $\bw(v,\parw)$ are globally Lipschitz and bounded in $\parw\in\setreal^{\nparw}$.
\end{rmk}

The augmented adaptive observer is given by
\begin{equation}	
	\label{eq:adaptive_observer_int}
	\begin{split}
	  	\dot{\yest} &= 
		\Phi(v,\hat{w},u)\hat{\parv} + \bv(v,\hat{w},u)
	  	+ (\gain I +\Psi_v P \Psi_v^\transp) 
	  	(v-\yest)
	  	\\
	  	\dot{\west} &= 
	  	A(v,\hat{\parw})\hat{w}+\bw(v,\hat{\parw})
	  	+ \Psi_w P \Psi_v^\transp (v-\hat{v})
	  	\\
	  	\col(\dot{\parvest},\dot{\parwest}) 
	  	&= \gamma P \Psi_v^\transp(v-\hat{v})
	\end{split}
\end{equation}
where $\gain>0$ is a constant gain, and the matrices $P$ and
\[  \Psi := \col(\Psi_v,\Psi_w)
\] evolve according to
\begin{subequations}
	\label{eq:adaptive_matrices_int}
	\begin{align}
	  	\label{eq:dPsi_internal}
		\dot{\Psi}
	  	 &= 
	  	 A_\Psi(t)
	  	\Psi + 
		\gain  B_\Psi(t)
		\\
		  	\label{eq:dP_internal} 
	  	\dot{P} &= \alpha P 
	  	+ \beta I  - P \, \Psi_v^\transp \Psi_v P
	\end{align}
\end{subequations}
Here, $\alpha>0$ and $\beta\ge0$ are constant hyperparameters, and the matrix functions in \eqref{eq:dPsi} are given by 
\begin{subequations}
\begin{equation*}
	\label{eq:A_Psi} 
	A_\Psi(t) = \begin{bmatrix}
			- \gain I & 
			\partial_{\hat{w}} [\Phi(v,\hat{w},u)\sat_\parv(\hat{\parv})+\bv(v,\hat{w},u)]\\
			0_{\nw\times \nv} & A(v,\hat{\parw})
	  	 \end{bmatrix}
\end{equation*}
and
\begin{equation*}
	B_\Psi(t) = 		  	
		  	\begin{bmatrix}
	  		\Phi(v,\hat{w},u) & 0_{\nv\times \nparw} \\
	  		0_{\nw\times \nparv} &  \partial_{\hat{\parw}}
		[A(v,\hat{\parw})\sat_w(\hat{w})+\bw(v,\hat{\parw})]
	  	\end{bmatrix}
\end{equation*}
\end{subequations}
with $\varsigma_\parv$ and $\varsigma_\parw$ smooth saturation functions (see \cref{rmk:saturation_w}).
Here, we assume without loss of generality that $\Psi_v(0)=0$, $\Psi_w(0)=0$, and $P(0) \succ 0$. 

As before we need a persistent excitation condition:
\begin{assum}
\label{assum:persistent_excitation_int} 
The signals $v(t)$ and $u(t)$ are such that for any trajectory of \eqref{eq:adaptive_observer_int}, the matrix $\Psi_v(t)^\transp$ is persistently exciting.
\end{assum}

The following result now parallels \cref{lem:P}:
\begin{lem}
\label{lem:P_int} 
Under 	\cref{assum:true_invariant_set,assum:int_dyn_contraction,assum:known_sets}
the solution $\Psi(t)$ of \eqref{eq:dPsi_internal} is bounded for all $t \ge 0$. In addition, under \cref{assum:persistent_excitation_int}, the solution $P(t)$ of \eqref{eq:dP_internal} is bounded and uniformly positive definite for all $t \ge 0$. In particular,
\begin{equation}
		\label{eq:P_bounds_beta}
		0 \prec \Pmin I \preceq P(t) \preceq \Pmax I
\end{equation}
for all $t \ge T$, with
\begin{equation}
	\label{eq:pbounds_beta}
	\begin{split}
	\Pmin  &= \left(\|P^{-1}(0)\|+ \alpha^{-1} \, \bar{c}^2 \right)^{-1} \\
	\Pmax &= \delta^{-1} e^{2\alpha T} 
(1 + \beta \, \delta^{-1} \alpha^{-3} e^{2\alpha T} \bar{c}^4)
	\end{split}
\end{equation}
with $\bar{c}$ a constant independent of $\alpha$, $\beta$, and $\gamma$.
\end{lem}
\begin{proof}
	See \cref{proof:P_int}.
\end{proof}

To state our main result, we gather the variables of \eqref{eq:true_system} in
\begin{equation}
	\label{eq:x}
	x(t) := \col(v(t),w(t),\parv(t),\parw(t))
\end{equation}
and the variables of \eqref{eq:adaptive_observer_int} in 
\begin{equation}
	\label{eq:xest}
	\xest(t) := \col(\vest(t),\west(t),\parvest(t),\parwest(t))
\end{equation}
We shall prove the result with contraction analysis, using the contraction metric given by 
\begin{equation}
	\label{eq:M}
	M(t) = T(t)^\transp\bar{M}(t) T(t)
\end{equation}
with
\begin{equation}
	\label{eq:T_int} 
	T =
	\left[
	\begin{array}{c|c}
	I & - \displaystyle\frac{\Psi}{\gamma} \\ \hline
	0 & I
	\end{array}
	\right],\;
	\bar{M} = 
	\left[
	\begin{array}{cc|c}
	\varepsilon I & 0 & 0 \\
	0 & \Pint & 0 \\\hline
	0 & 0 & \varepsilon (\gamma P)^{-1}
	\end{array}
	\right],
\end{equation}
and $\varepsilon>0$ (we use lines to delimit block sub-matrices of the same size). 
We use this contraction metric to define the set
\begin{equation}
\label{eq:ball} 
\	\mathcal{B}(t\,;\rho) = \left\{z \,:\, \left\|z-\hat{x}(t)\right\|_{M(t)} \le \rho \right\}
\end{equation}
for arbitrary $\rho>0$. 
We now have:

\begin{thm}
\label{thm:convergence_internal} 
Let \cref{assum:true_invariant_set,assum:int_dyn_contraction,assum:known_sets,assum:persistent_excitation_int} hold, and let $\alpha>0$, $\beta \ge 0$, and $\gamma>0$. Then there exist $\Mmin,\Mmax >0$ such that 
\begin{equation}
		\label{eq:M_bounds}
		0 \prec \Mmin I \preceq M(t) \preceq \Mmax I
\end{equation}
for all $t \ge 0$, and there is a constant $r > 0$ such that if
\begin{equation}
	\label{eq:initial_condition} 
	x(0) \in \mathcal{B}(0\,;r\sqrt{\Mmin}),
\end{equation}
then 
$x(t) \in \mathcal{B}(t\,;r\sqrt{\Mmin})$ for all $ t \ge 0$. 
Furthermore, if \eqref{eq:initial_condition} holds, $\xest(t) \to x(t)$ as $t \to 0$, exponentially fast, with rate $\ratevirt < \min\{\alpha,\rateint,\gamma\}$.

\end{thm}
\begin{proof}
See \cref{proof:convergence_internal}.
\end{proof}

\begin{rmk}
The observer above extends the simpler observer \eqref{eq:adaptive_observer}-\eqref{eq:adaptive_matrices}. 
	Indeed, if
	$A$ and $\bw$ are independent of $\parw$, then we recover \eqref{eq:adaptive_observer}-\eqref{eq:adaptive_matrices} from \eqref{eq:adaptive_observer_int}-\eqref{eq:adaptive_matrices_int}: in this case, 
	$\partial_{\hat{\parw}}[A(v)\sat_w(\hat{w})+\bw(v)] = 0$, 
	and $\dot{\Psi}_w = A(v)\Psi_w$ in \eqref{eq:dPsi}. Thus, if $\Psi_w(0) = 0$ we have $\Psi_w(t) = 0$ for all $t \ge 0$ and \eqref{eq:adaptive_observer}-\eqref{eq:adaptive_matrices} is recovered. If $\Psi_w(0) \neq 0$, then $\Psi_w \to 0$ as $t \to \infty$ by \cref{assum:int_dyn_contraction}, and the simpler observer is also recovered.
\end{rmk}


\begin{rmk}
	\label{assum:Psi_w2} 
	Writing $\Psi_w = \begin{bmatrix}
	\Psi_{w,1} & \Psi_{w,2} \end{bmatrix}$ with $\Psi_{w,1}\in\setreal^{\nw\times\nparv}$ and $\Psi_{w,2}\in\setreal^{\nw\times\nparw}$, by \cref{assum:int_dyn_contraction} we have without loss of generality that $\Psi_{w,1} = 0$. Furthermore, \cref{assum:persistent_excitation_int} is in a sense also a condition on the excitation of $\Psi_{w,2}$. To see this, write $\Psi_v = \begin{bmatrix} \Psi_{v,1} & \Psi_{v,2} \end{bmatrix}$ analogously. Then from \eqref{eq:dPsi_internal} the dynamics of $\Psi_{v,1}$ are solely driven by $\Phi(v,\hat{w},u)$, while the dynamics of $\Psi_{v,2}$ are solely driven by $\partial_{\hat{w}} [\Phi(v,\hat{w},u)\sat_\parv(\hat{\parv})+\bv(v,\hat{w},u)]\Psi_{w,2}$. Thus part of the persistent excitation of $\Psi_v$ is directly due to $\Psi_{w,2}$.
	
	
%
\end{rmk}

\subsection{Measurement errors}
\label{sec:noise}

The adaptive observer design of \cref{sec:advanced_observer} relies on \textit{output injection}, that is, it assumes that the measurement $y = v$ has no errors, and injects the measured $v$ in the observer dynamics. This corresponds to an \textit{equation error} model structure \cite{ljung_system_1999}. Assume instead that a measurement error $e(t)$ is present, so that 
\begin{equation}
	\label{eq:measurement_errors} 
	y(t) = v(t) + e(t)
\end{equation}
In this case, \eqref{eq:adaptive_observer_int}-\eqref{eq:adaptive_matrices_int} must be redefined by replacing $v(t)$ with $y(t)$. It is clear that this introduces measurement errors in the observer dynamics, and it is well known that even if some level of stability is retained, the parameter estimates will be biased. 
If the bias is too large, one could modify \eqref{eq:adaptive_observer_int}-\eqref{eq:adaptive_matrices_int} towards an \textit{output error} model structure by replacing $v(t)$ with $\vest(t)$ in the arguments of $\Phi$, $A$, $a$, and $b$. 
The downside of the output-error approach 
is that convergence of the adaptive observer may be lost even when the measurements have no errors (the \textit{nominal} case $y=v$). More precisely, a convergence analysis analogous to that of \cref{thm:convergence_internal} shows that nominal stability of the output error--based adaptive observer only holds for sufficiently high values of $\gamma$. But a high $\gamma$ is undesirable when measurement errors do occur, as $\gamma$ contributes to perturbations in the dynamics coming from $\gamma I e$ and $\gamma P \Psi_v^\transp e$.

To leverage the advantages of both equation error and output error approaches (lower gain $\gamma$ and lower bias, respectively), we can exploit an additional property of the true system, motivated by the neuronal systems of \cref{sec:cb_models}:
\begin{assum}
	\label{assum:contraction_v} 
	Under \cref{assum:true_invariant_set}, there exist a symmetric positive definite matrix $\Pv \succ 0$
	and a contraction rate $\ratev > 0$ such that 
	\begin{equation*}
		\partial_v[\Phi \parv + \bv]^\transp \Pv 
		+ \Pv \partial_v[\Phi \parv + \bv]
		\preceq -\ratev \Pv
	\end{equation*}
	for all $\{v,w,\parv\} \in V \times W \times \{\parv(0)\}$.
\end{assum}
\begin{rmk}
For any conductance-based model from \cref{sec:cb_models}, \cref{assum:contraction_v} holds with $\Pv = I$ and $\ratev = -2\maxcond_\Leak/c$. 
\end{rmk}

\cref{assum:contraction_v} motivates the observer structure given by
\begin{equation*}	
	\begin{split}
	  	\dot{\yest} &= 
		\Phi(\vest,\hat{w},u)\hat{\parv} + \bv(\vest,\hat{w},u)
	  	+(  \gain I +\Psi_v P \Psi_v^\transp) 
	  	(y-\yest)
	  	\\
	  	\dot{\west} &= 
	  	A(y,\hat{\parw})\hat{w}+\bw(y,\hat{\parw})
	  	+ \Psi_w P \Psi_v^\transp (y-\hat{v})
	  	\\
	  	\col(\dot{\parvest},\dot{\parwest}) 
	  	&= \gamma P \Psi_v^\transp(y-\hat{v})
	\end{split}
\end{equation*}
and by \eqref{eq:adaptive_matrices_int}, where $A_\Psi$ and $B_\Psi$ are replaced by
\begin{subequations}
	\label{eq:Psi_dynamics} 
\begin{equation*}
	A_\Psi(t) = \begin{bmatrix}
			- \gamma I + \partial_{\vest}[\hat{\Phi} \sat_\parv(\hat{\parv}) + \hat{\bv}] & 
			\partial_{\hat{w}} [\hat{\Phi}\sat_\parv(\hat{\parv})+\hat{\bv}]\\
			0_{\nw\times \nv} & A(\vest,\hat{\parw})
	  	 \end{bmatrix}
\end{equation*}
and
\begin{equation*}
	B_\Psi(t) = 
		  	\begin{bmatrix}
	  		\hat{\Phi} & 0_{\nv\times \nparw} \\
	  		0_{\nw\times \nparv} & \partial_{\hat{\parw}}
		[A(y,\hat{\parw})\sat_w(\hat{w})+\bw(y,\hat{\parw})]
	  	\end{bmatrix}
\end{equation*}
\end{subequations}
where $\hat{\Phi} = \Phi(\vest,\west,u)$ and $\hat{\bv} = \bv(\vest,\west,u)$.
The following nominal convergence result is immediate:
\begin{thm}
\label{thm:output_error} 
Under \cref{assum:contraction_v}, for $y=v$, the statement of \cref{thm:convergence_internal} also applies to the adaptive observer above.
\end{thm}
\begin{proof}
The proof follows the very same steps as that of \cref{thm:convergence_internal}, and is hence omitted.
\end{proof}

\subsubsection{Robustness to measurement errors}

The contraction results of \cref{thm:convergence_linear,thm:convergence_internal,thm:output_error}
imply a nominal exponential stability property of the adaptive observer trajectories. Given those results, it is not difficult to show that the convergence properties of the adaptive observers presented above all have some level of robustness with respect to measurement errors of the form \eqref{eq:measurement_errors}. A contraction-based robustness analysis along the lines of \cite[Section III]{bonnabel_contraction_2015} can be performed to show that for a bounded error $e(t)$ and  sufficiently small $\sup_{t\ge0} \|e(t)\|$, the trajectories of the adaptive observer estimates remain close to the trajectories of the true system states. 

\subsection{Robustness to unstructured uncertainty}
\label{sec:robustness}

It is well known that with exponential contraction comes robustness with respect to small perturbations \cite{lohmiller_contraction_1998,bonnabel_contraction_2015}.  In our context, consider a perturbed version of the true system \eqref{eq:true_system}, given by 
 \begin{equation}
 \label{eq:internal_perturbed} 
 \begin{split}
  \dot{v} &= \Phi(v,w,u) \parv + \bv(v,w,u) + d_v(t,v,w,\parv)
  \\
  \dot{w} &= A(v,\eta)w + \bw(v,\eta) + d_w(t,v,w) \\
  \dot{\parv} &= d_{\parv}(t,v,w) \\
  \dot{\parw} &= d_{\parw}(t,v,w)
  \end{split}
 \end{equation}
where $d:=\col(d_v,d_w,d_\parv,d_\parw)$ models an unstructured uncertainty. The disturbances can be interpreted as \textit{model mismatch} resulting from unmodelled dynamics, as well as time variation in the true parameters. 
We assume that the assumptions of \cref{thm:convergence_internal} hold for the perturbed system \eqref{eq:internal_perturbed}, with the set $\{\theta(0)\}\times\{\eta(0)\}$ in \cref{assum:true_invariant_set} replaced by $\Theta \times H$ from \cref{assum:known_sets}, and that $\|d\|\le\bar{d}$ for all $t \ge 0$.
We gather the state variables of \eqref{eq:internal_perturbed} and \eqref{eq:adaptive_observer_int} in 
$x=\col(v,w,\parv,\parw)$ and $\xest=\col(\hat v, \hat w, \hat \parv, \hat \parw)$, respectively. 

 In the proof of \cref{thm:convergence_internal}, we have shown that there exists an $r>0$ such that the set $\mathcal{B}(t\,; r\sqrt{\Mmin})$ given by \eqref{eq:ball} is contained at all times in a region of contraction with respect to a virtual system containing the trajectories of the \textit{nominal} true system \eqref{eq:true_system} and of the adaptive observer \eqref{eq:adaptive_observer_int}. In this situation, just as in \cite[Section III]{bonnabel_contraction_2015}, we can show that if the  state $x(0)$ of the \textit{perturbed} true system \eqref{eq:internal_perturbed} belongs to $\mathcal{B}(0\,; r\sqrt{\Mmin})$, then
\begin{equation}
	\label{eq:virtual_error_inequality} 
		\|x(t)-\xest(t)\|
		\le \sqrt{\frac{\Mmax}{\Mmin}}\left(e^{-\frac{\ratevirt}{2} t}  
		\|x(0)-\xest(0)\| 
		+ \frac{2}{\lambda} \bar{d}
		\right)
\end{equation}
for as long as the perturbed state $x(t)$ remains in the contraction region (see also \cite{del_vecchio_contraction_2013}). For small enough $\bar{d}$, this holds for all $t \ge 0$, as the state $x(t)$ will remain in $\mathcal{B}(t\,; r\sqrt{\Mmin})$. 

The maximum permissible $\bar{d}$ can hence be understood as a robustness margin. This (possibly conservative) margin allows for a simple interpretation of the effects that the hyperparameters will have on the robustness of the algorithm. 
Notice from \eqref{eq:M}-\eqref{eq:T_int}, we have
\[
\frac{\Mmax}{\Mmin}
\le \frac{\overline{\mu}}{\underline{\mu}}
\frac{\sup_{t \ge 0}\|T(t)\|^2}{\inf_{t \ge 0}\sigma_{\min}[T(t)]^2}
\]
where $\underline{\mu} I \le \bar{M}(t) \le \overline{\mu}I$.
Then, from \eqref{eq:P_bounds_beta} and \eqref{eq:T_int}, 
for $t \ge T$ we have
\begin{equation}
\label{eq:mus} 
	\frac{\overline{\mu}}{\underline{\mu}} = 
	\frac{\max\{\varepsilon,1,\varepsilon(\gamma\underline{p})^{-1}\}}{\min\{\varepsilon,\lambda_{\min}[M_w],\varepsilon(\gamma\overline{p})^{-1}
	\}}
\end{equation}
where $\underline{p}$ and $\overline{p}$ are given by \eqref{eq:pbounds_beta} and $\varepsilon$ is given by \eqref{eq:epsilon_int}. 
The maximum permissible $\bar{d}$ must decrease when $\overline{\mu}/\underline{\mu}$ increases. Hence, considering a \textit{fixed} contraction rate $\lambda<\min\{\alpha,\lambda_w,\gamma\}$, we can extract a few points about the values of the hyperparameters.

First, large values of $\alpha$ are undesirable, since $\overline{p}$ increases exponentially with $\alpha$. Thus there is a tradeoff between quickly keeping track of time-varying parameters and robust stability. Second, since the quantities $\varepsilon$ and $\varepsilon/\gamma$ are monotonically increasing in $\gamma$, $\gamma$ can be increased (up to a point) to enhance robustness. Additionally, notice that the proof of \cref{thm:convergence_internal} shows through \eqref{eq:sufficient_contraction} that increasing $\gamma$ allows for a larger $r>0$, which also promotes robustness.

Notice these points also apply to the simple adaptive observer \eqref{eq:adaptive_observer}, with the caveat that the virtual system was in that case proven to be globally contracting, and hence \eqref{eq:virtual_error_inequality} applies for any $\bar{d}$. Finally, although we have not analysed the effect of $\beta\ge 0$ above, it should be mentioned that the addition of a moderate $\beta I > 0$ term in \eqref{eq:dP_internal} can be regarded as a form of \textit{covariance inflation}, which has been noted to improve in practice the robustness of Kalman filter-based methods \cite{ghobadi_robust_2018}. This effect  was observed in the simulations that follow.

\section{Application to conductance-based models}
\label{sec:numeric} 

In this section we illustrate with numerical simulations how the system theoretic adaptive observers discussed in this paper perform when applied to problems in electrophysiology\footnote{The Julia code used to generate these results can be found on \texttt{https://github.com/thiagoburghi/online-learning}.}.

\subsection{Estimation of voltage and ion channel dynamics}
\label{sec:HH_numeric}

The primary goal of neuronal system identification is 
to estimate 
capacitances and maximal conductances  \cite{huys_efficient_2006,druckmann_novel_2007,meliza_estimating_2014,burghi_feedback_2021}.  An important (and often ignored) point in this approach is the fact that the parameters in ionic channel models are only approximate in nature, and in practice may vary from neuron to neuron. Using the adaptive observer \eqref{eq:adaptive_observer_int}-\eqref{eq:adaptive_matrices_int}, in this section we illustrate a real-time solution to the problem of estimating the unknown and uncertain parameters of the Hodgkin-Huxley model of \cref{ex:HH_description,ex:HH_parametrization}.

\subsubsection{Perfect measurements}

We begin by verifying the behaviour of the adaptive observer when no measurement errors are present. 
We use the biophysical
parameters in \cref{sec:HH_parameters}. This results in 
\[
	\begin{split}
		\parv(t) &= \parv(0) = \begin{pmatrix}
			1,&120,&36,&0.3
		\end{pmatrix}^\transp \\
		\parw(t) &= \parw(0) = \begin{pmatrix}
			-40,&-62,&-53
		\end{pmatrix}^\transp
	\end{split}
\]
By contrast, we initialize the observer with
\[
	\begin{split}
		\hat{\parv}(0) &= \begin{pmatrix}
			2,&78,&78,&10
		\end{pmatrix}^\transp \\
		\hat{\parw}(0) &= \begin{pmatrix}
			-20,&-20,&-20
		\end{pmatrix}^\transp
	\end{split}
\]
which represents a parsimonious guess over the parameters of an unknown and uncertain spiking conductance-based model.
The true HH model and the 
adaptive observer were simulated subject to the input 
$
	\Iapp_\app(t) = \sin(2\pi t / 10)
$ 
for $t \ge 0$. 
The initial conditions of the voltage and gating variables are given by $\col(v(0),w(0)) = (-30,0.5,0.5,0.5)^\transp$ and $\col(\vest(0),\west(0)) = (-30,0,0,0)^\transp$
and the remaining initial conditions of the observer
are given by $\Psi_v(0) = 0$, $\Psi_w(0) = 0$, $P(0) = I$.
For $\alpha = 0.1$ and $\beta = \gamma = 1$, the solutions of the
true system and of the 
adaptive observer can be seen in Figure \ref{fig:HH_parameters}. 
All the parameter estimates of the adaptive observer converge to
the true parameter values.

\begin{figure}
	\centering
	\if\usetikz1
	\tikzset{png export}
	\begin{tikzpicture}
		\def\plotdomain{250}
		\begin{groupplot}[
			group style={group size=1 by 7,
 			vertical sep=20pt,
			},
			height=3.5cm,width=6.5cm,
			axis y line = left,
			axis x line = bottom,	
			tick label style={font=\scriptsize},	
			label style={font=\scriptsize},
			legend style={font=\footnotesize},
			xmax=300,
			legend pos=outer north east,
			xmax = \plotdomain
			]
			\nextgroupplot[
				ymax=7.5
				]
			\addplot[dashed,blue!30!black,domain=0:\plotdomain,semithick]{1};
				\addlegendentry{$\parv=1/c$};
			\addplot[color=blue,semithick]
				table[x index=0,y index=4] 
				{./data/HH_parameters.txt};
				\addlegendentry{$\parvest_1=1/\hat{c}$}; 
			\addplot[dashed,red!30!black,domain=0:\plotdomain,semithick]{0.3};
				\addlegendentry{$\parv_4=\maxcond_{\Leak}/c$};
			\addplot[color=red,semithick]
				table[x index=0,y index=3] 
				{./data/HH_parameters.txt};
				\addlegendentry{$\parvest_4=\hat{\maxcond}_{\Leak}/\hat{c}$}; 
				
			\nextgroupplot[
				ymax = 130,
				ymin = 26
				]
			\addplot[dashed,blue!30!black,domain=0:\plotdomain,semithick]{120};
				\addlegendentry{$\parv_2=\maxcond_{\Na}/c$};
			\addplot[color=blue,semithick]
				table[x index=0,y index=1] 
				{./data/HH_parameters.txt};
				\addlegendentry{$\parvest_2=\hat{\maxcond}_{\Na}/\hat{c}$}; 
			\addplot[dashed,red!30!black,domain=0:\plotdomain,semithick]{36};
				\addlegendentry{$\parv_3=\maxcond_{\K}/c$};
			\addplot[color=red,semithick]
				table[x index=0,y index=2] 
				{./data/HH_parameters.txt};
				\addlegendentry{$\parvest_3=\hat{\maxcond}_{\K}/\hat{c}$}; 

			\nextgroupplot[
				ymax=-10,
				ymin = -65,
				xlabel={$t$ [ms]}, 
				]
			\addplot[dashed,blue!30!black,domain=0:\plotdomain,semithick]{-40};
				\addlegendentry{$\parw_1 = \vhalf_{m_\Na}$};
			\addplot[color=blue,semithick]
				table[x index=0,y index=5] 
				{./data/HH_parameters.txt};
				\addlegendentry{$\parwest_1 = \hat{\vhalf}_{m_\Na}$}; 
			\addplot[dashed,red!30!black,domain=0:\plotdomain,semithick]{-62};
				\addlegendentry{$\parw_2 = \vhalf_{h_\Na} $};
			\addplot[color=red,semithick]
				table[x index=0,y index=6] 
				{./data/HH_parameters.txt};
				\addlegendentry{$\parwest_2 = \hat{\vhalf}_{h_\Na}$}; 
			\addplot[dashed,green!30!black,domain=0:\plotdomain,semithick]{-53};
				\addlegendentry{$\parw_3 = \vhalf_{m_\K}$};
			\addplot[color=green,semithick]
				table[x index=0,y index=7] 
				{./data/HH_parameters.txt};
				\addlegendentry{$\parwest_3 = \hat{\vhalf}_{m_\K}$}; 
		\end{groupplot}
	\end{tikzpicture}
	\else
		\includegraphics[scale=0.125]{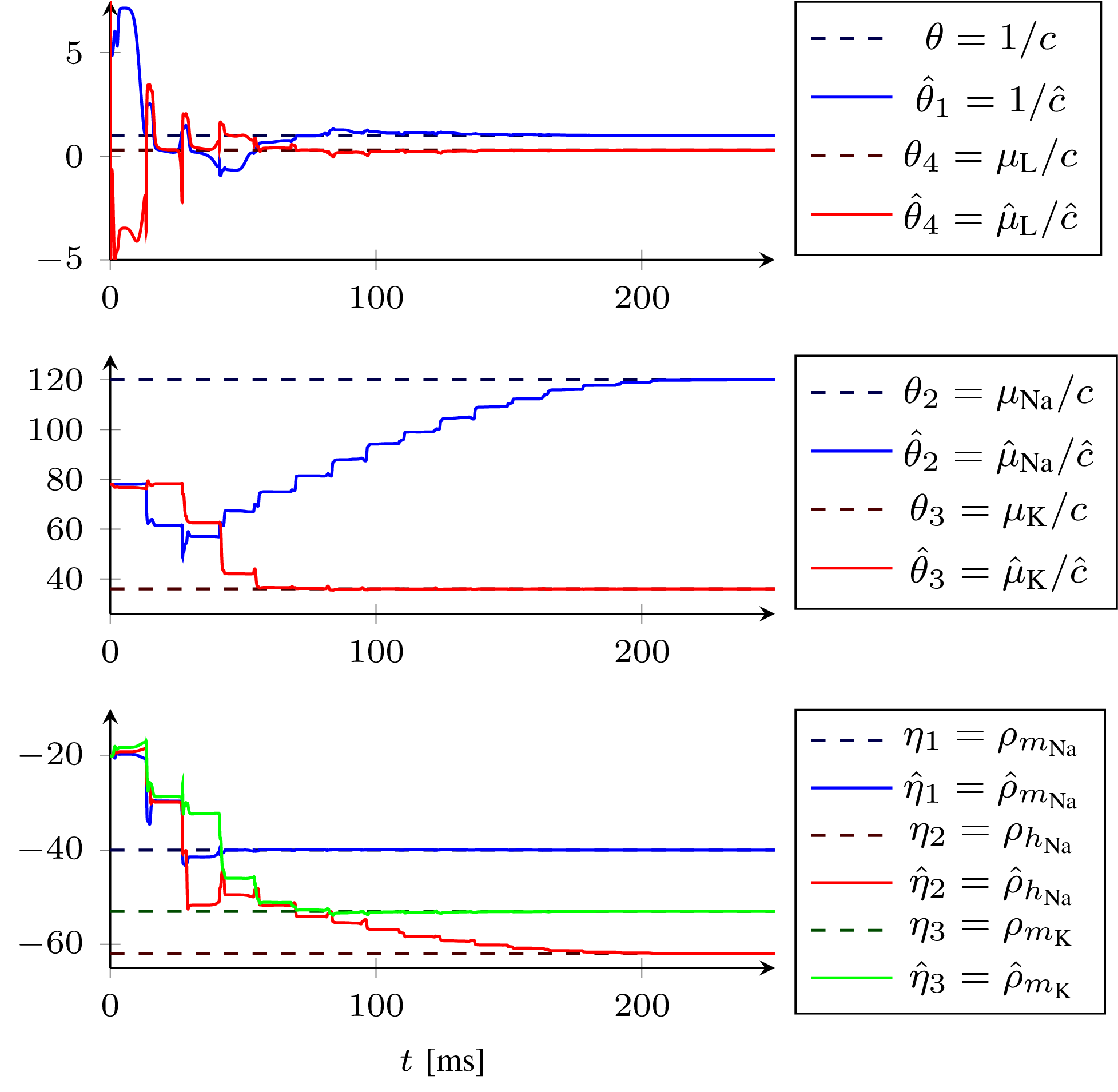}
	\fi
	\caption{Estimated parameters of the adaptive observer \eqref{eq:adaptive_observer_int}-\eqref{eq:adaptive_matrices_int} in the estimation of the HH model when no measurement errors are present. Here, $\alpha = 0.1$ and $\beta = \gamma = 1$.}
	\label{fig:HH_parameters}
\end{figure}
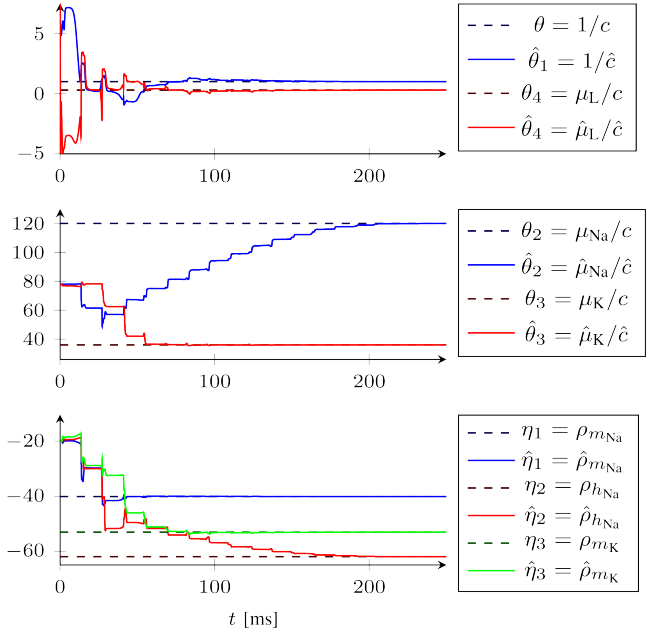

\subsubsection{Measurement errors}

Keeping the same input and true system parameters used in the previous section, we now simulate the behaviour of the adaptive observer when zero-mean white Gaussian noise of variance $\sigma^2_{\textit{noise}}=4 \; \mathrm{mV^2}$ is added to the measured $v$. Using the rms amplitude $v_{\textit{rms}}\approx 27$ mV of the noise-free voltage trace simulated in the previous section, this noise corresponds to a relatively poor signal-to-noise-ratio of $10\log_{10} v_{\textit{rms}}^2/\sigma_{\textit{noise}}^2 \approx 22$ dB. We first try using the same observer parameters $\alpha = 0.1$ and $\beta = \gamma = 1$ that previously led to convergence in the previous section. The result for the worst affected estimate, $\hat{\parv}_2$, is shown in Figure \ref{fig:HH_parameters_noise} (top). While the estimate remains bounded, it can be seen that the measurement noise considerably affects its convergence properties. However, tuning the observer parameters to $\alpha = 10^{-4}$, $\beta = 10$, $\gamma = 0.1$ and $P(0) = 0.1 I$ drastically improves the result, as shown in \ref{fig:HH_parameters_noise} (bottom). The oscillation in $\hat{\parv}_2(t)$ is now much less pronounced, and it converges more slowly to a region close to the true $\parv_2$. Similar behaviours hold for the less affected parameters. Comparing the two cases, there is a clear tradeoff between convergence rate and robustness to measurement noise.

\begin{figure}
	\centering
	\if\usetikz1
	\tikzset{png export}
	\begin{tikzpicture}
		\begin{groupplot}[
			group style={group size=1 by 4,
			 			vertical sep=20pt,},
			height=3.0cm,width=6.5cm, 
			axis y line = left, 
			axis x line = bottom,	
			tick label style={font=\scriptsize},	
			label style={font=\scriptsize},
			legend style={font=\footnotesize},
			legend pos=outer north east,
			x filter/.code=\pgfmathparse{\pgfmathresult/1000},
			]
			\nextgroupplot[
				ymax = 220,
				ymin = 50,
				ytick={80,120,160}
				]
				\addplot[dashed,blue!30!black,domain=0:9910,semithick]{120};
				\addlegendentry{$\parv_2=\maxcond_{\Na}/c$};
				\addplot[smooth,color=blue,semithick]
				table[x index=0,y index=1] 
				{./data/HH_parameters_a=0.1_g=1.0_s=2_oi=true.txt};
				\addlegendentry{$\parvest_2=\hat{\maxcond}_{\Na}/\hat{c}$}; 

			\nextgroupplot[
				ymax = 170,
				ymin = 50,
				ytick={80,120,160},
				xlabel={$t$ [s]}
				]
				\addplot[dashed,blue!30!black,domain=0:9910,semithick]{120};
				\addlegendentry{$\parv_2=\maxcond_{\Na}/c$};
				\addplot[smooth,color=blue,semithick]
				table[x index=0,y index=1] 
				{./data/HH_parameters_a=0.0001_g=0.1_s=2_oi=true.txt};
				\addlegendentry{$\parvest_2=\hat{\maxcond}_{\Na}/\hat{c}$};		
		\end{groupplot}
	\end{tikzpicture}
	\else
		\includegraphics[scale=0.125]{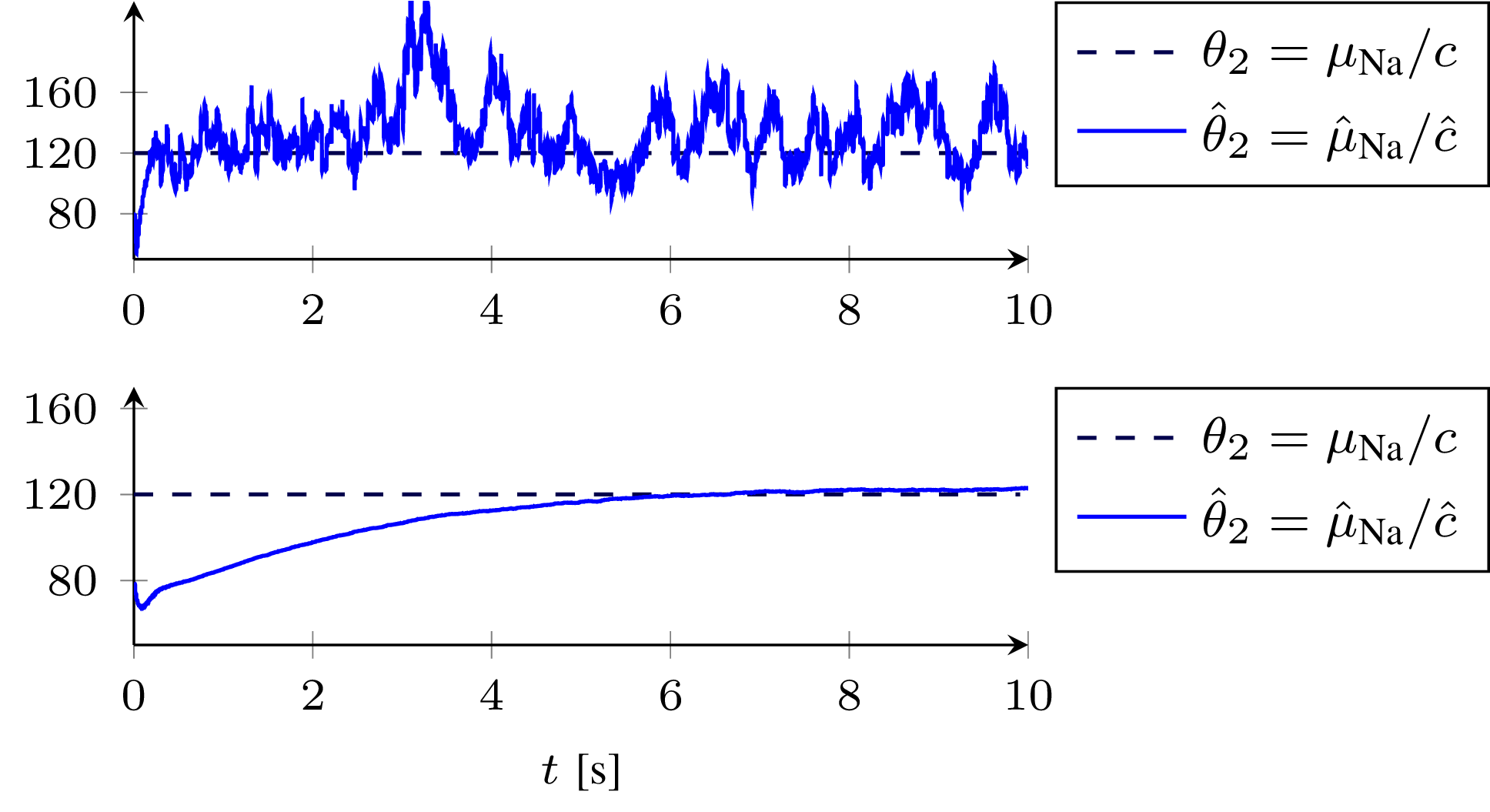}
	\fi
	\caption{Estimated parameters of the adaptive observer \eqref{eq:adaptive_observer_int}-\eqref{eq:adaptive_matrices_int} in the estimation of the HH model when a white noise measurement error is present ($\sigma_{\textit{noise}} = 2$ mV). \textbf{Top:} $\alpha = 0.1$, $\beta = 1$, and $\gamma = 1$. \textbf{Bottom:} $\alpha = 10^{-4}$, $\beta = 10$, and $\gamma = 0.1$.}
	\label{fig:HH_parameters_noise}
\end{figure}
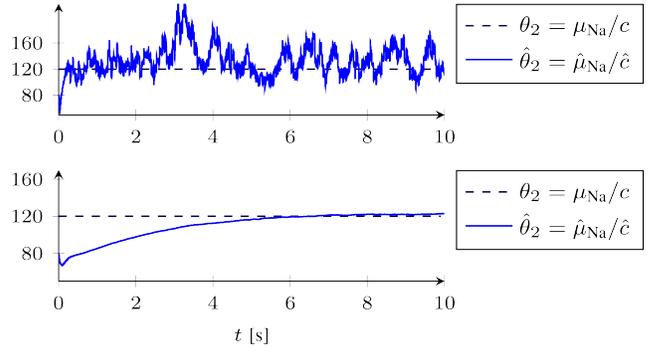

\subsection{Estimation of a neural circuit under neuromodulation}
\label{sec:HCO_neuromodulation}

In this section, we illustrate the robustness to noise, model mismatch, and time-varying parameters by using the observer of \cref{sec:noise} to estimate the parameters $\parv$ of the HCO neuronal circuit introduced in \cref{ex:HCO_description}.  

\begin{rmk}
\label{rmk:distributed} 
When applied to a conductance-based network, the network observers decouple into $\nv$ independent single neuron observers. This is because in a conductance-based network, $\Phi$ is block-diagonal and, by stability of the dynamics of $\Psi$ in \eqref{eq:dPsi} or \eqref{eq:dPsi_internal} and of $R:=P^{-1}$ in \eqref{eq:R_sol}, we can without loss of generality ignore all off-block diagonal terms of the matrices $\Psi(t)$ and $P(t)$.
\end{rmk} 

Following \cref{rmk:distributed}, applying the observer of \cref{sec:noise} to the HCO of \cref{ex:HCO_description} yields
	\begin{equation*}
		\begin{split}
		  	\dot{\vest}_i &= 
			\varphi_i(\vest_i,\west^{(i)})\maxcondest^{(i)}
			+ c_i^{-1} u_i
		  	+ (\gain + \psi_i 
		  	P_i \psi_i^\transp) 
		  	(y_i-\vest_i) 
		  	\\
		  	\dot{\west}^{(i)} &= 
			A_i(y)\hat{w}^{(i)} + \bw(y,\hat{w}^{(i)})
		  	\\
		  	\dot{\maxcondest}^{(i)} &= \gain P_i \, 
		  	\psi_i^\transp \, (y_i-\vest_i)
		  	\\
		  	\dot{\psi}_i &= 
			(- \gain I + \partial_{\vest_i}[\varphi_i(\vest_i,\west^{(i)})\varsigma(\maxcondest^{(i)})])\psi_i + \gain\varphi_i(\vest_i,\west^{(i)}),
		  	\\
		  	\dot{P}_i &= \alpha P_i + \beta I -  
		  	P_i \, \psi_i^\transp \psi_i P_i
		\end{split}
		\end{equation*}
		where $P_i(0) \succ 0$, $i\in\mathcal{N}=\{1,2\}$, and $y_i = v_i + e_i$. As in the previous section, we define the measurement errors $e_1$ and $e_2$ as white noise with $\sigma^2_{\textit{noise}} = 4\;\mathrm{mV^2}$ . 
	
\begin{figure}
	\centering
	\if\usetikz1
	\tikzset{png export}
	\begin{tikzpicture}
		\begin{groupplot}[
			group style={group size=1 by 3,
 			vertical sep=20pt,
			},
			height=2.5cm,width=8.0cm,
			axis y line = left,
			axis x line = bottom,	
			tick label style={font=\scriptsize},	
			label style={font=\scriptsize},
			legend style={font=\footnotesize},
			each nth point=2,
			filter discard warning=false,
			x filter/.code=\pgfmathparse{\pgfmathresult/1000},
			]
			\nextgroupplot[
				ylabel={$v_1(t) \; \mathrm{[mV]}$}]
			\addplot[color=blue]
				table[x index=0,y index=1] 
				{./data/HCO_truevoltages_seed=101pcerror=0.01.txt};
			\nextgroupplot[
				ylabel={$v_2(t) \; \mathrm{[mV]}$}]
			\addplot[color=blue]
				table[x index=0,y index=2] 
				{./data/HCO_truevoltages_seed=101pcerror=0.01.txt};
			\nextgroupplot[
				ylabel={$\mathrm{[\upmu A/cm^2]}$},
				legend pos = south east,
				xlabel={$t$ [s]}]
			\addplot[color=black,semithick]
				table[x index=0,y index=11]
				{./data/HCO_parameters_seed=101pcerror=0.01.txt};
				\addlegendentry{$\maxcond_{\Ca,1}=\maxcond_{\Ca,2}$}; 
		\end{groupplot}
	\end{tikzpicture}
	\else
		\includegraphics[scale=0.125]{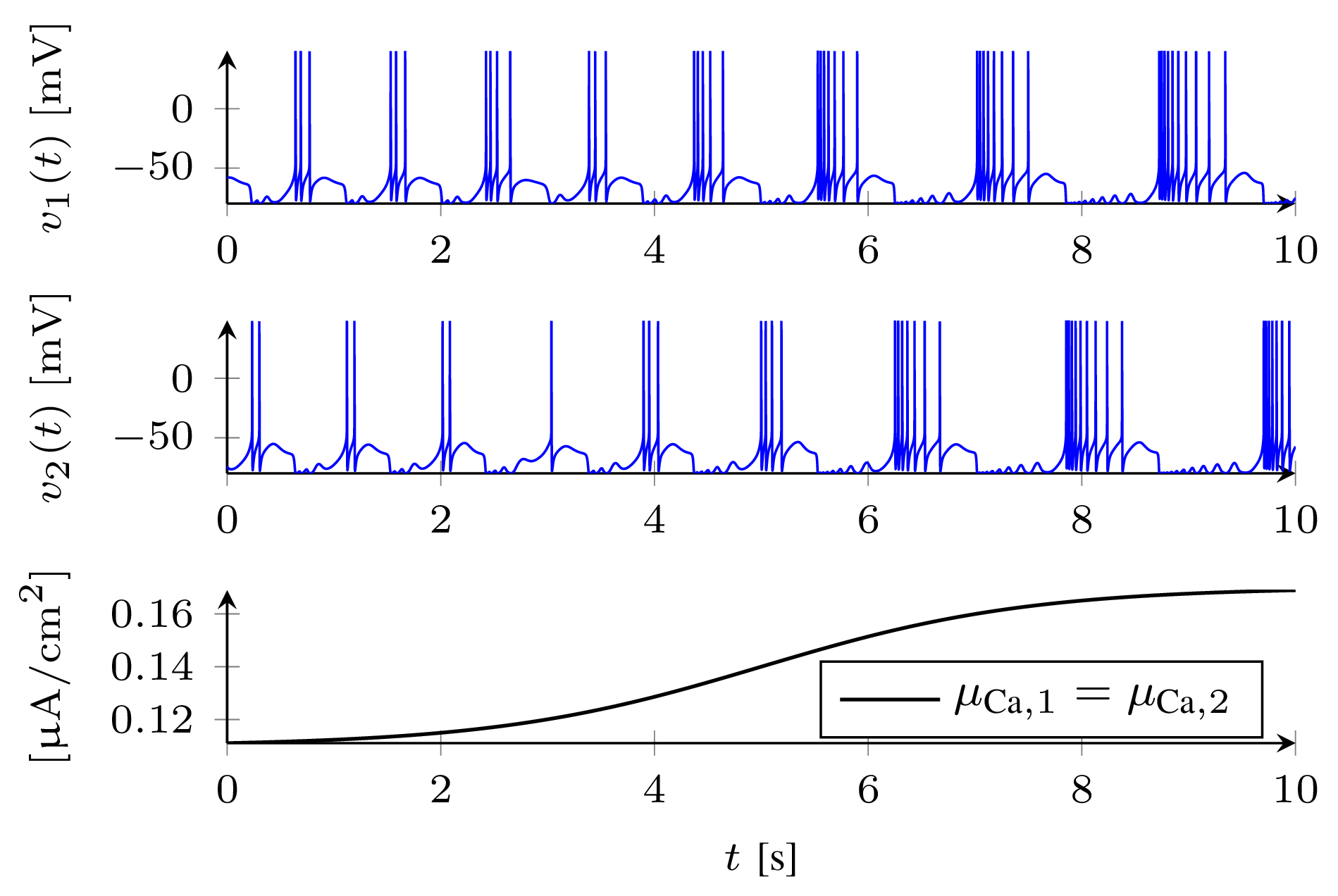}
	\fi
	\caption{True HCO voltage traces and time-varying calcium conductance used in \cref{sec:HCO_neuromodulation}.}
	\label{fig:HCO_truvoltages_neuromodulation}
\end{figure}
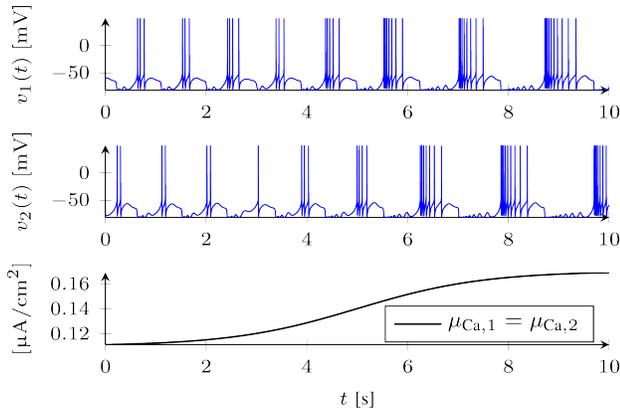

To illustrate the importance of tracking time-varying parameters, we consider the problem of neuromodulation \cite{marder_neuromodulation_2014}. Neuromodulators are substances that continuously modulate the opening of ion channels in a neuron's membrane. This modulatory control can be modelled as a temporal variation of the maximal conductances in a conductance-based model \cite{drion_cellular_2019}. Here, we consider the case in which the calcium maximal conductances $\maxcond_{\Ca,1}(t)$ and $\maxcond_{\Ca,2}(t)$ of the true HCO model are slowly varied in time, something that is known to change the bursting frequency of the HCO model \cite{dethier_positive_2015}.  A gradual increase in the concentration of calcium ion channels is simulated by 
\begin{equation}
	\label{eq:maxcond_Ca}
	 \maxcond_{\Ca(t),1}(t) = 
	 \maxcond_{\Ca,2}(t) = 
	 0.11 + \frac{0.07}{1+\mathrm{exp}\left(
	 	-\tfrac{t-T_f/2}{1250}\right)}
\end{equation}
where $T_f = 10$ seconds is the length of the simulation. 
For $i\in \{1,2\}$, the remaining maximal conductances of the true HCO model are given by $\maxcond_{\Na,i}=60$, $\maxcond_{\K,i}=40$, $\maxcond_{\Leak,i}=0.035$, and $\maxcond_{\GABA,2,1} = \maxcond_{\GABA,1,2} = 4$.

In the observer above, the reversal potentials,  capacitances, activation functions and time-constant functions are defined according to the nominal parameters of the true model detailed in \cref{sec:HCO_parameters} (where initial conditions are also detailed). To simulate an unknown disturbance $d_w$ in the true internal dynamics (see \cref{sec:robustness}), a random disturbance of at most $1\%$ (following the uniform distribution) is applied to every internal dynamics parameter of the true system.

Figure \ref{fig:HCO_truvoltages_neuromodulation} illustrates the resulting voltage traces of the true (perturbed) HCO model. The neuromodulatory action on the calcium conductance increases the number of spikes in each burst. For a forgetting rate of $\alpha = 0.0025$,  observer gains of $\beta = 0$ and $\gamma = 0.1$, and a constant input $\Iapp_1(t) = \Iapp_2(t) = -0.65$ $\mathrm{\upmu A/cm^2}$, Figure \ref{fig:HCO_parameters_neuromodulation} shows the trajectories of some of the true and estimated maximal conductances.  It can be seen that the estimates converge towards a region close to the true parameters, illustrating the robustness of the convergence property of the observer. The bias in the estimates after convergence is expected, as the internal dynamics of the observer and of the true model are different due to model mismatch (but simulating the estimated model with fixed $\parvest(t=10s)$ results in half-center oscillations congruent with the those of the true system). The calcium conductance estimates track the true calcium conductances by remaining in a time-varying region around the true value. A comparison with Figure \ref{fig:HCO_truvoltages_neuromodulation} shows that calcium estimates are corrected whenever a burst of spikes (a rich part of the signal) is elicited. 

\begin{figure}
	\centering
	\if\usetikz1
	\tikzset{png export}
	\begin{tikzpicture}
		\begin{groupplot}[
			group style={group size=1 by 3,
 			vertical sep=20pt,
			},
			height=3.5cm,width=6.5cm,
			axis y line = left, 
			axis x line = bottom,	
			tick label style={font=\scriptsize},	
			label style={font=\scriptsize},
			legend style={font=\footnotesize},
			filter discard warning=false,				
			legend pos=outer north east,
			yticklabel style={
			        /pgf/number format/fixed,
			        /pgf/number format/precision=2
			},
			scaled y ticks=false
			]
			\nextgroupplot[
				ylabel={$\mathrm{[mS/cm^2]}$},
				ymin=20,
				ymax=100,
				]
			\addplot[smooth,color=blue,semithick,each nth point=5]
				table[x expr=\thisrowno{0}/1000,y index=1] 
				{./data/HCO_parameters_seed=101pcerror=0.01.txt};
				\addlegendentry{$\maxcondest_{\Na,1}$}; 
			\addplot [smooth,color=red,semithick,each nth point=5]
				table[x expr=\thisrowno{0}/1000,y index=2] 
				{./data/HCO_parameters_seed=101pcerror=0.01.txt};
				\addlegendentry{$\maxcondest_{\Na,2}$};
			\addplot[dashed,black,domain=0:10,semithick]{60};
				\addlegendentry{$\maxcond_{\Na,i}$};
				
			\addplot[smooth,color=magenta,semithick,each nth point=5]
				table[x expr=\thisrowno{0}/1000,y index=3] 
				{./data/HCO_parameters_seed=101pcerror=0.01.txt};
				\addlegendentry{$\maxcondest_{\K,1}$}; 
			\addplot [smooth,color=cyan,semithick,each nth point=5]
				table[x expr=\thisrowno{0}/1000,y index=4] 
				{./data/HCO_parameters_seed=101pcerror=0.01.txt};
				\addlegendentry{$\maxcondest_{\K,2}$};
			\addplot[dashed,black,domain=0:10,semithick]{40};
				\addlegendentry{$\maxcond_{\K,i}$};
			\nextgroupplot[
				ylabel={$\mathrm{[mS/cm^2]}$},
				ymin = 0.0,
				ymax = 0.2,
				]
			\addplot [smooth,color=blue,semithick,each nth point=5]
				table[x expr=\thisrowno{0}/1000,y index=5] 
				{./data/HCO_parameters_seed=101pcerror=0.01.txt};
				\addlegendentry{$\maxcondest_{\Ca,1}$}; 
			\addplot[smooth,color=red,semithick,each nth point=5]
				table[x expr=\thisrowno{0}/1000,y index=6] 
				{./data/HCO_parameters_seed=101pcerror=0.01.txt};
				\addlegendentry{$\maxcondest_{\Ca,2}$};
			\addplot[dashed,color=black,semithick,each nth point=5]
				table[x expr=\thisrowno{0}/1000,y index=11] 
				{./data/HCO_parameters_seed=101pcerror=0.01.txt};
				\addlegendentry{$\maxcond_{\Ca,i}$};
			\nextgroupplot[
				ylabel={$\mathrm{[mS/cm^2]}$},
				ymin=0,
				xlabel={$t$ [s]}
				]
			\addplot[smooth,color=blue,semithick,each nth point=5]
				table[x expr=\thisrowno{0}/1000,y index=9] 
				{./data/HCO_parameters_seed=101pcerror=0.01.txt};
				\addlegendentry{$\maxcondest_{\GABA,2,1}$}; 
			\addplot[smooth,color=red,semithick,each nth point=5]
				table[x expr=\thisrowno{0}/1000,y index=10] 
				{./data/HCO_parameters_seed=101pcerror=0.01.txt};
				\addlegendentry{$\maxcondest_{\GABA,1,2}$}; 
			\addplot[dashed,black,domain=0:10,semithick]{4};
				\addlegendentry{$\maxcond_{\GABA,i,j}$}; 
		\end{groupplot}
	\end{tikzpicture}
	\else
		\includegraphics[scale=0.125]{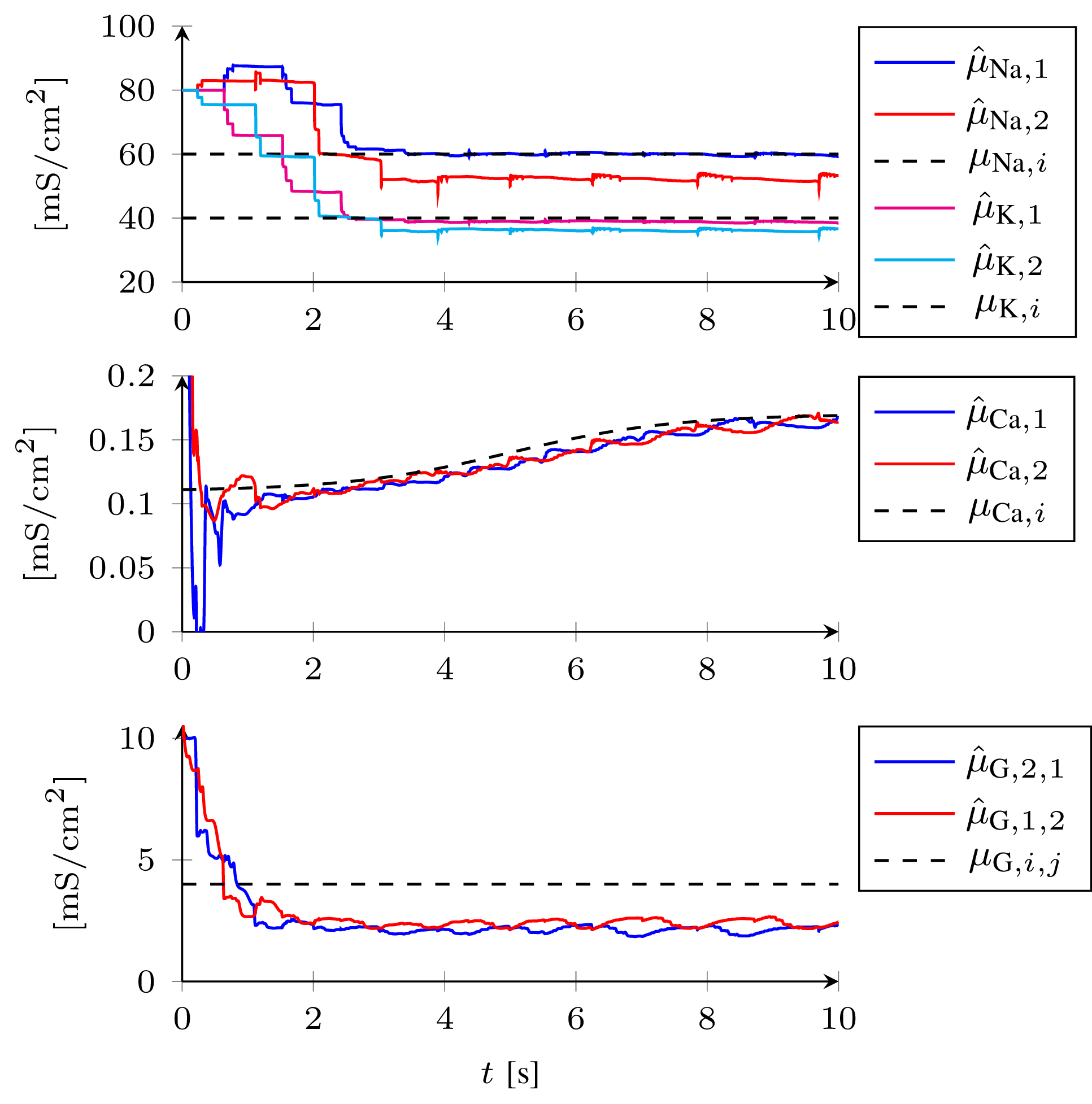}
	\fi
	\caption{Example trajectories of HCO maximal conductance estimates when noise and model mismatch are present (see \cref{sec:HCO_neuromodulation}).}
	\label{fig:HCO_parameters_neuromodulation}
\end{figure}
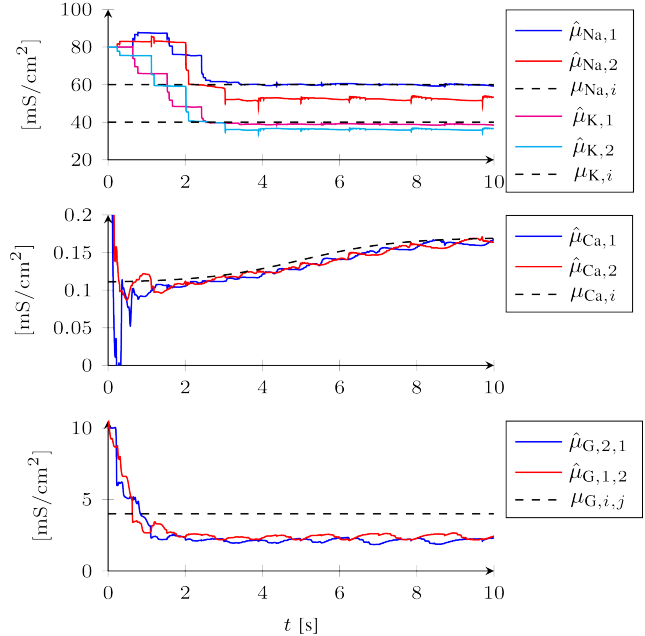

\section{Concluding remarks}

We have presented an adaptive observer which can be used for
real-time estimation of conductance-based models of neuronal circuits. 
The observers in this work can be used for indirect adaptive control of neuronal maximal conductances \cite{schmetterling_adaptive_2022}, opening the way for innovative neurophysiology research. Future work will explore the benefits and limitations of the the method in an experimental context.

\vspace{1em}
	
\section{Acknowledgements}

We thank Dr Fulvio Forni and Raphael Schmetterling for the feedback given during the writing of this manuscrupt. We also thank the anonymous reviewers who greatly contributed to the improvement of the results.
	
\appendix

\subsection{Contraction analysis}
\label{sec:contraction_analysis}
 
The system dynamics 
\begin{equation}
	\label{eq:nonlinear_system} 
	\dot{x} = f(x,u)
\end{equation}
is said to be \textit{exponentially contracting}
\cite{lohmiller_contraction_1998}
in $x$ on $X \subset \setreal^{n_x}$, uniformly in 
$u$ on $U \subseteq \setreal^{n_u}$, if there
exist a continuously differentiable symmetric
matrix $P(x,t)$, called the \textit{contraction
metric}, and a constant $\lambda > 0$, called the
\textit{contraction rate}, such that 
$\epsilon_1 I \preceq P(x,t)\preceq \epsilon_2 I$
for some $\epsilon_1,\epsilon_2>0$, and 
$\partial_x f^\transp P + P \partial_x f
+ \dot{P} \preceq -\lambda P$ for all $t \ge 0$,
all $x \in X$, and all $u \in U$. 
The set $X\subset \setreal^{n_x}$ is said to be
\textit{positively invariant} with respect to the
dynamics \eqref{eq:nonlinear_system}, uniformly in
$u$ on $U \subseteq \setreal^{n_u}$, if $x(0) \in X$
and $u(t) \in U$ for all $t \ge 0$ imply 
$x(t) \in X$ for all $t \ge 0$. It is a well-known
fact that if the dynamics 
\eqref{eq:nonlinear_system} are exponentially
contracting on a convex positively invariant set
$X$, then all solutions of that system starting in 
$X$ converge towards each other exponentially
fast, with rate $\lambda$ (for a proof of this 
statement, see for instance \cite[Lemma 1]{jouffroy_tutorial_2010}).


\subsection{Proofs}

\subsubsection{Proof of \cref{lem:invariant_set}}
\label{proof:invariant_set} 
	We begin by noticing that $[0,1]$ is a positively invariant set for \eqref{eq:activation} and for \eqref{eq:inactivation}, uniformly in $v$ on $\setreal$. This is because the image of the sigmoid \eqref{eq:sigmoid} is $(0,1)$, which implies none of the gating variables $m_\ion$ and $h_\ion$ can leave the set $[0,1]$: for instance, $\dot{m}_\ion \ge 0$ for $m_\ion = 0$ and all $v \in \setreal$, and $\dot{m}_\ion \le 0$ for $m_\ion = 1$ and all $v \in \setreal$. Now, assuming $m_\ion(0) \in [0,1]$ and $h_\ion(0) \in [0,1]$, we have $\maxcond_\ion m_\ion^{p_\ion}h_\ion^{q_\ion} > 0$ for all $\ion\in\mathcal{I}$ and all $t \ge 0$. This in turn implies $v$ cannot leave the interval $[\vmin,\vmax]$, which can be verified by inspection of  \eqref{eq:single_neuron_cb}-\eqref{eq:ion_currents}:  if $v = \vmax$, then $\dot{v} \le 0$, whereas if $v = \vmin$, then $\dot{v} \ge 0$.
	
\subsubsection{Proof of \cref{prop:rls}}
\label{proof:rls}

The normal equation of the LS problem \eqref{eq:batch_problem}-\eqref{eq:filtered} with $R_0(T) = e^{-\alpha T} P^{-1}(0)/T$ is 
\begin{equation*}
	\label{eq:normal_equation} 
	R(T)
	\hat{\parv}(T) 
	= 
	\textstyle\int_0^T e^{-\alpha(T-\tau)}\Psi(\tau)^\transp(H\dot{v}(\tau)-H\hat{a}(\tau))d\tau 
\end{equation*}
where
\begin{equation}
	\label{eq:R_sol} 
	R(t) = e^{-\alpha t} P^{-1}(0) +
	\textstyle\int_0^t e^{-\alpha(t-\tau)} \Psi(\tau)^\transp\Psi(\tau)d\tau
\end{equation}
Differentiating the normal equation by $T$ and evaluating at $t$ we obtain the RLS solution	
	\begin{equation*}
		\dot{\hat{\parv}}(t) =  
		P(t)\Psi(t)^\transp 
		\big(H\dot{v}(t) - \Psi(t) \hat{\parv}(t) - H\hat{\bv}(t)\big)
	\end{equation*}
	Thus \eqref{eq:adaptive_observer}-\eqref{eq:adaptive_matrices} implements the RLS solution if and only if
	\begin{equation}
		\label{eq:rls_equivalence} 
		H\dot{v}(t) - \Psi(t)\hat{\parv}(t) - H \hat{\bv}(t)
		= \gain(v(t)-\vest(t))	
	\end{equation}
	To verify the above identity, we 
	first notice that 
	\[
		\begin{split}
		\frac{d}{dt}(\Psi \hat{\parv})
		&= -\gain \Psi\hat{\parv} + \gain \Phi\hat{\parv}
		+\gain \Psi P \Psi^\transp(v-\vest)\\
		&= -\gain \Psi\hat{\parv} 
		+ \gain (\dot{\vest} - \gamma(v-\vest)- \hat{\bv})
		\end{split}
	\]
	Solving the previous equation for 
	$\Psi \hat{\parv}$, we obtain
	\[
		\Psi(t) \hat{\parv}(t)
		= -\gain Hv(t) - H\hat{\bv}(t)
		  + \gain \vest(t)
	\]
	We can now recover \eqref{eq:rls_equivalence}
	by adding $\gain(v(t)-\vest(t))$ to both
	sides of the previous equation and applying
	the identity 
	\[
		\gain(v(t) - Hv(t)) = H\dot{v}(t)
	\]
	which can be easily verified from
	\eqref{eq:filter}. 
	
\subsubsection{Proof of \cref{lem:P}}
\label{proof:P}
	Consider the system
	\begin{equation}
		\label{eq:dR} 
		\dot{R} = 
		-\alpha R + \Psi^\transp\Psi
	\end{equation}
	with $R(0) = P(0)^{-1} \succ 0$, whose solution is given by \eqref{eq:R_sol}. We claim that $R(t)$ is uniformly positive definite and bounded for all $t\ge 0$, and that
		$\Pmax^{-1} I \preceq R(t)
		\preceq \Pmin^{-1} I$,
for $t \ge T$ with the bounds given by \eqref{eq:kappa1}. In this case, \eqref{eq:P_bounds} follows from setting $P(t) = R(t)^{-1}$ and checking that the identity $\dot{P}=\dot{R}^{-1}= -R^{-1}\dot{R}R^{-1}$ leads to \eqref{eq:dP}. To prove the claim, we first notice $R(t) \ge e^{-\alpha T} R(0)$ for $0 \le t \le T$. For $t \ge T$, we can show that $R(t) \ge \Pmax^{-1}I$ by following the same steps as in the proof of \cite[Lemma 1]{zhang_adaptive_2001}. The upper bound $\Pmin^{-1}I$ of $R(t)$ can be obtained as follows: first, we notice that \eqref{eq:dPsi} yields 
$\|\Psi(t)\| \le \overline{\phi}$ for all $t\ge 0$. Then, since \eqref{eq:R_sol} is the solution to \eqref{eq:dR}, we have
\begin{align*}
\|R(t)\| 
&\le \|R(0)\|+ \alpha^{-1} \sup_{\tau \ge 0}\|\Psi(\tau)\|^2 \le \alpha^{-1}\overline{\phi}^2
\end{align*}
proving the claim.

\subsubsection{Proof of \cref{thm:convergence_linear}}
\label{proof:convergence_linear}
We prove this result using the \textit{virtual system} idea of contraction analysis \cite{lohmiller_contraction_1998,jouffroy_tutorial_2010}: we construct a so-called virtual system whose solutions contain the solutions of both \eqref{eq:true_system_simple} and \eqref{eq:adaptive_observer}; then we show that the virtual system is globally exponentially contracting; this will imply that any solutions of \eqref{eq:true_system_simple} and \eqref{eq:adaptive_observer} converge exponentially fast towards each other.
We consider the virtual state vector 
\[\xvirt = \col(\vvirt,\wvirt,\parvvirt)\] and the virtual system given by
\begin{equation}
	\label{eq:virtual_system}
	\begin{split}
		\dot{\vvirt} &= 
		\fvvirt(t,\wvirt,\parvvirt) + (\gain I+\Psi P \Psi^\transp)(v-\vvirt)  
		\\		
		\dot{\wvirt} &= 
		A(v)\wvirt + \bw(v)
		\\		
		\dot{\parvvirt} &= \gamma P \Psi^\transp (v - \vvirt)
	\end{split}
\end{equation}
where
\[
	\fvvirt(t,\wvirt,\parvvirt) = 
	\Phi(v,\wvirt,u)\theta +
	\Phi(v,\west,u)
	(\parvvirt - \theta)
	+ \bv(v,\wvirt,u)
\]
By construction of $\fvvirt(t,\wvirt,\parvvirt)$, any solutions $x=\col(v,w,\theta)$ of \eqref{eq:true_system_simple} and $\xest=\col(\vest,\west,\thetaest)$ of \eqref{eq:adaptive_observer} are particular solutions of the virtual system \eqref{eq:virtual_system}; notice that $v(t)$, $\west(t)$, and $u(t)$ are not states of the virtual system. 

To show that the virtual system is globally exponentially contracting, we use the differential Lyapunov function 
\begin{equation}
	\label{eq:partialV} 
	\delta V(t,\delta \xvirt) = \delta \xvirt^\transp T(t)^\transp \bar{M}(t) T(t) \delta \xvirt
\end{equation}
where
\begin{equation}
	\label{eq:T} 
	T =
	\begin{bmatrix}
	I & 0 & -\frac{\Psi}{\gain} \\
	0 & I & 0 \\
	0 & 0 & I
	\end{bmatrix},\quad
	\bar{M} = 
	\begin{bmatrix}
	\varepsilon I & 0 & 0 \\
	0 & \Pint & 0 \\
	0 & 0 & \varepsilon (\gamma P)^{-1}
	\end{bmatrix}
\end{equation}
and $\delta \xvirt$ is the state vector of the differential system $\dot{\delta \xvirt} =  J \delta \xvirt$, with
\begin{equation}
	\label{eq:jacobian}
	J   = 
	\begin{bmatrix}
		- (\gamma I+\Psi P \Psi^\transp) & \partial_{\wvirt}\fvvirt(t,\wvirt,\parvvirt) & \Phi(v,\hat{w},u) \\
		0 & 
		A(v)
		& 0 \\
		-\gamma P \Psi^\transp & 0 & 0
	\end{bmatrix}
\end{equation}
the Jacobian of the vector field of \eqref{eq:virtual_system}. It can easily be verified that
\begin{equation}
	\label{eq:dpartialV} 
	\dot{\delta V}(t,\delta \xvirt,\xvirt)
	= \delta \xvirt^\transp T(t)^\transp
	\left( \bar{J}^\transp \bar{M} + \bar{M} \bar{J} + \dot{\bar{M}} \right) T(t) \delta \xvirt
\end{equation}
where
\begin{equation}
	\label{eq:J_bar} 
	\bar{J} = (T J + \dot{T})T^{-1}
\end{equation}
Hence to show global contraction we must show that the metric
	\begin{equation}
		\label{eq:Pvirt} 
		\Pvirt(t) := T^\transp(t) \bar{M}(t) T(t) 
	\end{equation}	
	in \eqref{eq:partialV} is uniformly positive definite
	and bounded, and, in view of \eqref{eq:dpartialV}, that
	\begin{equation}
	\label{eq:contraction_ineq} 
	\bar{J}^\transp \bar{M} + \bar{M} \bar{J} + 	\dot{\bar{M}} \preceq -\lambda \bar{M}
\end{equation}
for all for all $\xvirt \in \setreal^{\nv+\nw+\nparv}$ and $t \ge 0$. Boundedness of $M(t)$ follows from \cref{lem:P}, which also ensures that $\bar{M}(t)$ is uniformly positive definite. Since $T(t)$ is uniformly full column rank, this implies $M(t)$ is uniformly positive definite. 


It remains to show \eqref{eq:contraction_ineq}.
Computing the left-hand side of \eqref{eq:contraction_ineq} from \eqref{eq:T}, \eqref{eq:jacobian}, and \eqref{eq:J_bar}, we obtain
\begin{align*}
	&\bar{J}^\transp \bar{M} + \bar{M} \bar{J} + \dot{\bar{M}} = 
	\\
	 &\begin{bmatrix}
 	-2\varepsilon\gamma I & 
 	\varepsilon \partial_{\wvirt}\fvvirt(t,\wvirt) & -\varepsilon\Psi \\
 	*
 	& {\scriptstyle A(v)^\transp\Pint+\Pint A(v)} & 0 \\
 	* & * & 
 	- \gain^{-1}\varepsilon(\Psi^\transp\Psi + \alpha P^{-1})
 \end{bmatrix}
 \preceq Q
\end{align*}
Where the upper bound matrix $Q$ is given by
\begin{equation*}
	\label{eq:Q} 
 Q =
 \begin{bmatrix}
 	-\varepsilon\gamma I & \varepsilon \partial_{\wvirt}\fvvirt(t,\wvirt) & 0 \\
 	*
 	& -\rateint \Pint & 0 \\
 	* & * & - \varepsilon\alpha (\gain P)^{-1}
 \end{bmatrix}
\end{equation*}
Finally, notice that 
\[
	\partial_{\wvirt}\fvvirt(t,\wvirt)
	= \partial_{\wvirt}\left(
	\Phi(v,\wvirt,u)\theta
	+ \bv(v,\wvirt,u)\right)
\]
is bounded (\cref{rmk:saturation_w}). Thus for any $\lambda < \min\{\alpha,\rateint,\gamma\}$ choosing
\begin{equation}
	\label{eq:epsilon} 
\varepsilon = (\rateint-\ratevirt) (\gamma-\ratevirt) \lambda_{\min}[\Pint] \sup \|\partial_{\wvirt} \tilde{f}(t,\wvirt)\|^{-2}
\end{equation}
ensures $Q \preceq -\lambda \bar{M}$ and hence \eqref{eq:contraction_ineq} holds globally, and the virtual system is globally exponentially contracting.

\subsubsection{Proof of \cref{lem:P_int}}
\label{proof:P_int} 

From  \cref{assum:true_invariant_set,assum:known_sets,rmk:saturation_parameters,rmk:saturation_w},
we have that the off-diagonal term in $A_\Psi(t)$ and the nonzero terms in $B_\Psi(t)$ are bounded for all $t \ge 0$. Let us define
$\overline{\phi} := \sup \|\Phi(v,w,u)\|$, 
$\overline{a} := \sup \|\partial_{\hat{w}}[\Phi(v,\hat{w},u)\varsigma_\theta(\hat{\theta}) + a(v,\hat{w},u)]\|$ and 
$\overline{b} = \max_i \sup \|\partial_{\hat{\eta_i}}[A(v,\hat{\eta})\varsigma_w(\hat{w})+b(v,\hat{\eta})]\|$ 
where the $\sup$s extend over $v \in V$, $u\in U$, $\hat{w}\in W$, $\hat{\theta} \in \Theta$ and $\hat{\eta} \in H$. 
To show $\Psi(t)$ is bounded, first denote each column of $\Psi_w(t)$ by $\psi_w^j(t)$, $j = 1,\dotsc,n_\parv+n_\parw$. Recalling that $\Psi_w(0)=0$, it follows from \eqref{eq:dPsi_internal} that $\psi^j_w(t) = 0$ for $j=1,\dotsc,n_\parv$ and all $t \ge 0$. Furthermore,  \cref{assum:int_dyn_contraction} implies that the dynamics $\dot{\psi}^j_w = A(v,\hat{\parw}) \psi_w^j$ is contracting, and hence it follows from \eqref{eq:dPsi_internal} that
\[
	\|\psi^j_w(t)\| \le \gain  \tfrac{2}{\lambda_w} \sqrt{\tfrac{\lambda_{\max}(M_w)}{\lambda_{\min}(M_w)}}\overline{b}
\]
for $j=n_\parv+1,\dotsc,n_\parv+n_\parw$ and all $t \ge 0$ \cite[Lemma 1]{del_vecchio_contraction_2013}. 
Since $\|\Psi_w(t)\| \le \nparw \max_j \|\psi_w^j(t)\|$ for each $t$, from \eqref{eq:dPsi_internal} we have
\[
	\|\Psi_v(t)\| \le \bar{c} := \overline{\phi} + \tfrac{2}{\lambda_w} \sqrt{\tfrac{\lambda_{\max}(M_w)}{\lambda_{\min}(M_w)}}\,\overline{a}\,\overline{b} \, \nparw
\]
for all $t \ge 0$, where $\bar{c}$ is independent of the hyperparameters $\alpha$, $\beta$, and $\gamma$. Hence $\Psi(t)$ is bounded.

To show that \eqref{eq:P_bounds_beta} holds, we first define the two systems
\begin{equation}
	\begin{split}
		\dot{\underline{P}} &= \alpha \underline{P} - \underline{P} \Psi_v^\transp\Psi_v \underline{P} \\
		\dot{\overline{P}} &= \alpha \overline{P} + \beta I
	\end{split}
\end{equation}
with $\overline{P}(0) = \underline{P}(0) = P(0)$. Then, by the the Comparison Theorem for the differential Riccati equations \cite[Theorem 4.1.4]{abou-kandil_matrix_2003}, it follows that
\[
	\underline{P}(t) \preceq P(t) \preceq \overline{P}(t)
\]
for all $t \ge 0$. Just as in \cref{lem:P}, we have 
$\underline{P}(t) \succeq \underline{p} I$ for all $t \ge 0$, with $\Pmin$ now given by \eqref{eq:pbounds_beta}. On the other hand, we have
\[
	\overline{P}(t) = e^{\alpha t}\overline{P}(0) + \tfrac{\beta}{\alpha}(e^{\alpha t}-1)I 
\]
for all $t \ge 0$ and hence $P(t)$ is upper bounded for $t \in [0,T)$. To find an upper bound for $P(t)$ for all $t \ge T$, we can use \cite[Lemma 2]{delyon_note_2001}, which specialized to our system states that the solutions of the system of equations
\begin{align*}
	\dot{R} &= -\alpha R + \Psi_v^\transp\Psi_v ,\quad & R(0) = 0 \\
	\dot{D} &= -\alpha D + \beta R^2, \quad & D(0) = 0
\end{align*}
satisfy
\begin{equation}
	\label{eq:upper_bound_P}
	P(t) \preceq R^{-1}(t) + R^{-1}(t)D(t)R^{-1}(t)
\end{equation}
as soon as $R^{-1}(t)$ exists. This is guaranteed from $t \ge T$, since, following the steps in the proof of \cite[Lemma 1]{zhang_adaptive_2001}, we have
\begin{equation}
	\label{eq:lower_bound_R}
	R(t) \succeq \delta e^{-2\alpha T} I
\end{equation}
for all $t \ge T$. Furthermore,
\[\|D(t)\| \le \frac{\beta}{\alpha} \sup_{t\ge 0}\|R(t)\|^2
\le \frac{\beta}{\alpha^3} \sup_{t\ge0}\|\Psi_v(t)\|^4 \le  \frac{\beta}{\alpha^3} \bar{c}^4 \]
for all $t \ge 0$. Hence it follows from \eqref{eq:upper_bound_P} and \eqref{eq:lower_bound_R} that
$P(t) \preceq \Pmax I$, for all $ t \ge T$, with $\Pmax$  given by \eqref{eq:pbounds_beta}. 

\subsubsection{Proof of \cref{thm:convergence_internal}}
\label{proof:convergence_internal} 

To begin, notice by \cref{lem:P_int}, the metric $M(t) = T(t)^\transp \bar{M}(t) T(t)$ given by \eqref{eq:T_int} is bounded, and $\bar{M}(t)$ is uniformly positive definite. Furthermore, since $T(t)$ is uniformly full column rank, $M(t)$ is also uniformly positive definite. This proves \eqref{eq:M_bounds}.

In the rest of the proof, in similar fashion to the proof of \cref{thm:convergence_linear}, we consider the virtual state vector 
\[\xvirt:=\col(\vvirt,\wvirt,\parvvirt,\parwvirt)\] 
and the virtual system
\begin{equation}
	\label{eq:virtual_system_int}
	\begin{split}
		\dot{\vvirt} &= 
		\fvvirt(t,\wvirt,\parvvirt) + (\gain I+\Psi_v P \Psi_v^\transp)(v-\vvirt)  
		\\		
		\dot{\wvirt} &= 
		\fwvirt(t,\wvirt,\parvvirt)	+ \Psi_w P \Psi_v^\transp (v-\vvirt)
		\\		
		\col(\dot{\parvvirt},\dot{\parwvirt}) 
		 &= \gamma P \Psi_v^\transp (v - \vvirt)
	\end{split}
\end{equation}
where
\[
	\begin{split}
	\fvvirt(t,\wvirt,\parvvirt) &= 
	\Phi(v,\wvirt,u)\parv + \Phi(v,\hat{w},u)(\parvvirt-\parv) + \bv(v,\wvirt,u)
	\\
	\fwvirt(t,\wvirt,\parwvirt) &=
	A(v,\parwvirt)w + A(v,\hat{\parw})(\wvirt-w) + \bw(v,\parwvirt)	
	\end{split}
\]
By construction of $\fvvirt$ and $\fwvirt$, any solutions $x=\col(v,w,\parv,\parw)$ of \eqref{eq:true_system} and $\xest=\col(\vest,\west,\hat{\parv},\hat{\parw})$ of \eqref{eq:adaptive_observer_int} are particular solutions of the virtual system \eqref{eq:virtual_system_int}. 

We use the differential Lyapunov equation 
$\partial V = \partial \xvirt^\transp T^\transp \bar{M} T \partial \xvirt := \partial \xvirt^\transp \Pvirt \partial \xvirt$, with $T$ and $\bar{M}$  given by \eqref{eq:T_int} and $\delta \xvirt$ the state of the differential system 
$\dot{\delta \xvirt} =  J \delta\xvirt$. The Jacobian $J$ of the vector field of \eqref{eq:virtual_system_int} is given by
\begin{equation}
\label{eq:jacobian_int} 
	J   = 
		\left[
			\begin{array}{cc|c}
			- \Psi P \Psi_v^\transp & 0 & 0  \\ \hline
			-\gamma P \Psi_v^\transp & 0 & 0 
			\end{array}
		\right]
		+ 
		\left[
			\begin{array}{c|c}
			\tilde{A}_\Psi(t) & \tilde{B}_\Psi(t) \\\hline 
			0 & 0 
			\end{array}
		\right]
\end{equation}
where
\begin{align*}
	\tilde{A}_\Psi(t) &= 
	\begin{bmatrix}
			- \gamma I & 
			\partial_{\wvirt} [\Phi(v,\wvirt,u)\parv+\bv(v,\wvirt,u)]\\
			0_{\nw\times\nv} & A(v,\hat{\parw})
	\end{bmatrix}
	\\
	\tilde{B}_\Psi(t) &=
	\begin{bmatrix}
		\Phi(v,\hat{w},u) & 0_{\nv\times\nparw} \\
		0_{\nw\times\nparv} & \partial_{\parwvirt}
		[A(v,\parwvirt)w+\bw(v,\parwvirt)]
	\end{bmatrix}
\end{align*}
(recall the lines delimit block sub-matrices of the same size).

As in the previous proof, $\partial\dot{V}(t,\partial\xvirt,\xvirt)$ satisfies \eqref{eq:dpartialV}, with $\bar{J}$ given by \eqref{eq:J_bar} but $J$ now given by \eqref{eq:jacobian_int} and $T$ now given by \eqref{eq:T_int}. Computing $\bar{J}$ while replacing $\dot{\Psi}$ by \eqref{eq:dPsi_internal}, we obtain 
\begin{equation*}
	\bar{J} =  \left[
		\begin{array}{cc|c}
		& &  		  
		\\
		\multicolumn{2}{c|}{\smash{\raisebox{.5\normalbaselineskip}{$\tilde{A}_\Psi(t)$}}}
		& 
		\smash{\raisebox{.5\normalbaselineskip}{$\left(\tilde{A}_\Psi(t)-A_\Psi(t)\right) \tfrac{\Psi}{\gain} + \tilde{B}_\Psi(t)-B_\Psi(t)$}}
		 \\ \hline
		-\gamma P \Psi_v^\transp & 0 & -P\Psi_v^\transp \Psi_v
		\end{array}
	\right]
\end{equation*}

Since the metric $M(t)$ is bounded and uniformly positive definite and since $\partial\dot{V}(t,\partial\xvirt,\xvirt)$ satisfies \eqref{eq:dpartialV}, to prove our result, we will find a contraction region in the state space where \eqref{eq:contraction_ineq} holds for all $t \ge 0$. Computing the left-hand side of \eqref{eq:contraction_ineq} from \eqref{eq:T_int} and $\bar{J}$ above, we obtain
\begin{equation*}
\bar{J}^\transp \bar{M} + \bar{M} \bar{J} + \dot{\bar{M}}
\preceq Q	
\end{equation*}
where the upper bound matrix $Q=Q^\transp$ is given by
\begin{equation*}
	Q =
	\left[
	\begin{array}{cc|cc}
	-\gamma\varepsilon I & 
	\varepsilon\partial_{\wvirt} \fvvirt(t,\wvirt) &
	 0 &
	 \varepsilon \gamma^{-1} \Delta_1 \Psi_{w,2}
	\\
	* & 
	-\rateint \Pint &
	0 & 
	\Pint \Delta_2 \\ \hline
	\multicolumn{2}{c|}{*} & 
	\multicolumn{2}{c}{
	-\varepsilon \gamma^{-1}\left( \alpha P^{-1} + \beta P^{-2} \right)
	} 
	\end{array}
	\right]
\end{equation*}
where we used $\Psi_w = \begin{bmatrix} 0 & \Psi_{w,2} \end{bmatrix}$ (see \cref{assum:Psi_w2}). Here, 
{\small
\begin{align*}
	\Delta_1 
	&= \partial_{\wvirt} [\Phi(v,\wvirt,u)\parv+\bv(v,\wvirt,u)]
		-\partial_{\hat{w}} [\Phi(v,\hat{w},u)\parv+\bv(v,\hat{w},u)]\\
		&\quad +\partial_{\hat{w}}[\Phi(v,\hat{w},u)(\parv-\sat_\parv(\hat{\parv}))]
\end{align*}
}
and
\begin{align*}
	\Delta_2
		&= \partial_{\parwvirt}[A(v,\parwvirt)w+\bw(v,\parwvirt)]
		-\partial_{\hat{\parw}}[A(v,\hat{\parw})w+\bw(v,\hat{\parw})] \\
		&\quad + \partial_{\hat{\parw}}[A(v,\hat{\parw})(w-\sat_w(\hat{w}))]
\end{align*}

We now wish to find a region of the state space where 
$Q \preceq -\ratevirt \bar{M}$ for all $t \ge 0$. For that purpose, let $\lambda < \min\{\alpha,\rateint,\gamma\}$, and consider an arbitrary number $\zeta \in (0,1)$. Then using Schur's complement we can show that the choice
\begin{equation}
	\label{eq:epsilon_int} 
\varepsilon = (1-\zeta)^2(\rateint-\ratevirt) (\gamma-\ratevirt) \lambda_{\min}[\Pint] \sup \|\partial_{\wvirt} \tilde{f}(t,\wvirt)\|^{-2}
\end{equation}
leads to
{\small
\begin{align*}
	&-Q -\ratevirt \bar{M} \succeq \\
	&\left[
	\begin{array}{cc|cc}
	\varepsilon\zeta(\gain-\ratevirt) I & 0 & 0 & -\varepsilon \gain^{-1} \Delta_1 \Psi_{w,2}
	\\
	* & \zeta(\rateint-\ratevirt) \Pint & 0 & -\Pint \Delta_2 
	\\ \hline
	\multicolumn{2}{c|}{*} & 
	\multicolumn{2}{c}{
	\varepsilon \gain^{-1} \left( (\alpha-\ratevirt) P^{-1} + \beta P^{-2} \right)
	} 
	\end{array}
	\right]
\end{align*}}

Hence it follows from Schur's complement that the right-hand side of the inequality above is positive semidefinite if and only if
\begin{equation*}
	\frac{\varepsilon}{\gamma} \left( (\alpha-\ratevirt) P^{-1} + \beta P^{-2} \right) \succeq 
	\begin{bmatrix}
	0 & 0 \\ 0 & \frac{\varepsilon}{\gain^2}\frac{\|\Delta_1\Psi_{w,2}\|^2}{\zeta(\gamma-\ratevirt)} + 
	\frac{\|\Delta_2\|_M^2}{\zeta(\rateint-\ratevirt)}
	\end{bmatrix}
\end{equation*}
By \cref{lem:P_int}, the above will hold if
\begin{equation}
	\label{eq:sufficient_contraction} 
(\alpha-\ratevirt)\Pmax^{-1}	+ \beta \Pmax^{-2}
	\ge 
	\frac{1}{\gamma}\frac{\|\Delta_1\Psi_{w,2}\|^2}{\zeta(\gamma-\ratevirt)}
	+ 
\frac{\gain}	{\varepsilon}
	\frac{\|\Delta_2\|_M^2}{\zeta(\rateint-\ratevirt)}
\end{equation}
where we recall $\zeta\in(0,1)$ is arbitrary.

By the continuity and global Lipschitz properties of $\Delta_1$ and $\Delta_2$, as well as by boundedness of $\Psi_{w,2}$, there exists a sufficiently small $r>0$ such that for each $t \ge 0$,  whenever $\|x(t)-\hat{x}(t)\| \le r$, the inequality \eqref{eq:sufficient_contraction} holds for all $\tilde{x}\in\setreal^{\nv+\nw+\nparv+\nparw}$ such that $\|\tilde{x} - \hat{x}(t)\| \le r$. Thus, as long as $\|x(t)-\hat{x}(t)\| \le r$ for all $t \ge 0$, the ball $\mathcal{B}(t;r \sqrt{\Mmin})$ given by 
\eqref{eq:ball} is contained in a region of contraction at all times, which implies that any trajectory of \eqref{eq:virtual_system_int} starting in $\mathcal{B}_r(0)$ remains in $\mathcal{B}_r(t)$ for $t \ge 0$ and converges exponentially fast to $\hat{x}(t)$ with rate $\lambda$ (see Theorem 2 of \cite{lohmiller_contraction_1998}).
But since $x(t)$ is a valid trajectory of \eqref{eq:virtual_system_int}, it follows that if $x(0) \in \mathcal{B}_r(0)$, then $x(t)$ remains in $\mathcal{B}_r(t)$, and exponential convergence of $x(t)$ to $\hat{x}(t)$ is guaranteed.

\subsection{Model parameters}
\label{sec:model_parameters} 

\subsubsection{Hodgkin-Huxley model}
\label{sec:HH_parameters} 
HH model parameters were adapted from \cite[pp. 46-47]{izhikevich_dynamical_2007}. Voltage dynamics parameters are given as follows:
	\begin{table}[h]
		\normalsize
		\centering
		\begin{tabular}{|p{.8cm}|p{.8cm}|p{.8cm}|p{.8cm}|
			p{.8cm}|p{1.0cm}|p{.6cm}|}
		\hline
		$\maxcond_\Na$ &	$\maxcond_\K$ & $\maxcond_{\rm{L}}$ &
		$\nernst_\Na$ & $\nernst_\K$ & $\nernst_{\rm{L}}$ &
		$c$ 
		\\
		\hline
		$120$ & $36$ & $0.3$ & $55$ & $-77$ & $-54.4$ & $1$ 
		\\
		\hline
		\end{tabular}
	\end{table}

All activation functions 
are of the form \eqref{eq:sigmoid}, and all time-constant 
functions  
are of the form \eqref{eq:gaussian_function},
with parameters given in the table below:

\begin{table}[h]
	\normalsize
	\centering
	\begin{tabular}{|p{.8cm}|p{.8cm}|p{.8cm}|p{.8cm}|
		p{.8cm}|p{.8cm}|p{.8cm}|}
	\hline
	& $\vhalf$ & $\slope$ & 
	$\taumin$ &
	$\taumax$ & 
	$\mean$ & 
	$\std$
	\\
	\hline
	$m_\Na$ & $-40$ & $9$ & $0.04$ & $0.50$ & $-38$ & $30$
	\\
	\hline			
	$h_\Na$ & $-62$ & $-7$ & $1.2$ & $8.6$ & $-67$ & $20$
	\\
	\hline		
	$m_\K$ & $-53$ & $15$ & $1.1$ & $5.8$ & $-79$ & $50$
	\\
	\hline
	\end{tabular}
	
	\label{tab:HH_internal_parameters} 
\end{table}

\subsubsection{Half-center oscillator}
\label{sec:HCO_parameters}

Both neurons in the HCO of \cref{sec:HCO_neuromodulation} have identical nominal capacitances, reversal potentials and internal dynamics. The reversal potentials are given by $\nernst_\Na=50$, $\nernst_\K=\nernst_\GABA=-80$, $\nernst_\Ca=120$, and $\nernst_\Leak=-49$ mV; the capacitances are given by $c_1 = c_2 = 1$.
The internal dynamics were adapted from \cite[p.2474]{dethier_positive_2015}. All activation functions are of the form \eqref{eq:sigmoid}, and all intrinsic time-constant functions are of the form \eqref{eq:gaussian_function},
with parameters given in the table below. The synaptic time-constant \eqref{eq:synaptic_time-constant} has $a_\GABA = 2$ and $b_\GABA = 0.1$.

\begin{table}[H]
	\normalsize
	\centering
	\begin{tabular}{|p{.6cm}|p{.8cm}|p{.8cm}|p{.8cm}|
		p{.8cm}|p{1.1cm}|p{1.0cm}|}
	\hline
	& $\vhalf$ & $\slope$ & 
	$\taumin$ &
	$\taumax$ & 
	$\mean$ & 
	$\std$
	\\
	\hline
	$m_\Na$ & $-35.5$ & $5.29$ & $0.06$ & $42.37$ & $-387.92$ & $133.78$
	\\
	\hline			
	$h_\Na$ & $-48.9$ & $-5.18$ & $1.50$ & $2.50$ & $-62.90$ & $10.00$
	\\
	\hline		
	$m_\K$ & $-12.3$ & $11.8$ & $0.80$ & $6.65$ & $-76.62$ & $61.42$
	\\
	\hline		
	$m_\Ca$ & $-67.1$ & $7.20$ & $1.01$ & $40.03$ & $-117.58$ & $62.87$
	\\
	\hline		
	$h_\Ca$ & $-82.1$ & $-5.5$ & $40.49$ & $126.51$ & $-92.48$ & $-50.24$
	\\
	\hline		
	$s_\GABA$ & $-45$ & $2$ & $-$ & $-$ & $-$ & $-$
	\\
	\hline
	\end{tabular}
	\label{tab:HCO_internal_parameters} 
\end{table}

We have chosen the HCO initial conditions $v(0)$ and $w(0)$ from the trajectory observed at steady-state oscillations with $\maxcond_{\Ca,1} = \maxcond_{\Ca,2} = 0.11$.
The adaptive observer initial conditions were arbitrarily set to
$\vest(0) = (-50,-50)^\transp$,
$\west^{(1)}(0) = \west^{(2)}(0) = 0$,
$\maxcondest_{\Na,1}(0) =
\maxcondest_{\Na,2}(0) = 80$,
$\maxcondest_{\K,1}(0) = 
\maxcondest_{\K,2}(0) = 80$,
$\maxcondest_{\Ca,1}(0) = 
\maxcondest_{\Ca,2}(0) = 1$,
$\maxcondest_{\Leak,1}(0) =
\maxcondest_{\Leak,2}(0) = 1$,
$\maxcondest_{\GABA,2,1}(0) =
\maxcondest_{\GABA,1,2}(0) = 10$,
$\psi^{(1)}(0) = \psi^{(2)}(0) = 0$,
and
$P^{(1)}(0) = P^{(2)}(0) = 0.1 I$.

\bibliographystyle{plain}
\bibliography{references}

\begin{thebibliography}{10}

\bibitem{abarbanel_estimation_2008}
Henry D.~I. Abarbanel, Daniel~R. Creveling, and James~M. Jeanne.
\newblock Estimation of parameters in nonlinear systems using balanced
  synchronization.
\newblock {\em Physical Review E}, 77(1), January 2008.

\bibitem{abou-kandil_matrix_2003}
Hisham Abou-Kandil, Gerhard Freiling, Vlad Ionescu, and Gerhard Jank.
\newblock {\em Matrix {Riccati} {Equations} in {Control} and {Systems}
  {Theory}}.
\newblock Birkhäuser, Basel, Switzerland, 2003.

\bibitem{besancon_remarks_2000}
Gildas Besançon.
\newblock Remarks on nonlinear adaptive observer design.
\newblock {\em Systems \& Control Letters}, 41(4):271--280, November 2000.

\bibitem{bonnabel_contraction_2015}
Silvere Bonnabel and Jean-Jacques Slotine.
\newblock A {Contraction} {Theory}-{Based} {Analysis} of the {Stability} of the
  {Deterministic} {Extended} {Kalman} {Filter}.
\newblock {\em IEEE Transactions on Automatic Control}, 60(2):565--569,
  February 2015.

\bibitem{burghi_feedback_2021}
Thiago~B. Burghi, Maarten Schoukens, and Rodolphe Sepulchre.
\newblock Feedback identification of conductance-based models.
\newblock {\em Automatica}, 123:109297, January 2021.

\bibitem{del_vecchio_contraction_2013}
Domitilla Del~Vecchio and Jean-Jacques~E. Slotine.
\newblock A {Contraction} {Theory} {Approach} to {Singularly} {Perturbed}
  {Systems}.
\newblock {\em IEEE Transactions on Automatic Control}, 58(3):752--757, March
  2013.

\bibitem{delyon_note_2001}
B.~Delyon.
\newblock A note on uniform observability.
\newblock {\em IEEE Transactions on Automatic Control}, 46(8):1326--1327,
  August 2001.
\newblock Conference Name: IEEE Transactions on Automatic Control.

\bibitem{dethier_positive_2015}
Julie Dethier, Guillaume Drion, Alessio Franci, and Rodolphe Sepulchre.
\newblock A positive feedback at the cellular level promotes robustness and
  modulation at the circuit level.
\newblock {\em Journal of Neurophysiology}, 114(4):2472--84, October 2015.

\bibitem{drion_neuronal_2015}
G.~Drion, T.~O'Leary, J.~Dethier, A.~Franci, and R.~Sepulchre.
\newblock Neuronal behaviors: {A} control perspective.
\newblock In {\em 54th {IEEE} {Conference} on {Decision} and {Control}}, pages
  1923--1944, December 2015.

\bibitem{drion_cellular_2019}
Guillaume Drion, Alessio Franci, and Rodolphe Sepulchre.
\newblock Cellular switches orchestrate rhythmic circuits.
\newblock {\em Biological Cybernetics}, 113(1):71--82, April 2019.

\bibitem{druckmann_novel_2007}
Shaul Druckmann, Yoav Banitt, Albert Gidon, Felix Schürmann, Henry Markram,
  and Idan Segev.
\newblock A {Novel} {Multiple} {Objective} {Optimization} {Framework} for
  {Constraining} {Conductance}-{Based} {Neuron} {Models} by {Experimental}
  {Data}.
\newblock {\em Frontiers in Neuroscience}, 1(1):7--18, October 2007.

\bibitem{ermentrout_mathematical_2010}
G.~Bard Ermentrout and David~H. Terman.
\newblock {\em Mathematical {Foundations} of {Neuroscience}}.
\newblock Springer, New York, 2010.

\bibitem{farza_adaptive_2009}
M.~Farza, M.~M’Saad, T.~Maatoug, and M.~Kamoun.
\newblock Adaptive observers for nonlinearly parameterized class of nonlinear
  systems.
\newblock {\em Automatica}, 45(10):2292--2299, October 2009.

\bibitem{gauthier_simple_1992}
J.P. Gauthier, H.~Hammouri, and S.~Othman.
\newblock A simple observer for nonlinear systems applications to bioreactors.
\newblock {\em IEEE Transactions on Automatic Control}, 37(6):875--880, June
  1992.

\bibitem{ghobadi_robust_2018}
Mostafa Ghobadi, Puneet Singla, and Ehsan~T. Esfahani.
\newblock Robust {Attitude} {Estimation} from {Uncertain} {Observations} of
  {Inertial} {Sensors} {Using} {Covariance} {Inflated} {Multiplicative}
  {Extended} {Kalman} {Filter}.
\newblock {\em IEEE Transactions on Instrumentation and Measurement},
  67(1):209--217, 2018.

\bibitem{hille_ionic_1984}
Bertil Hille.
\newblock {\em Ionic channels of excitable membranes}.
\newblock Sinauer Associates, Sunderland, MA, 1984.

\bibitem{hodgkin_quantitative_1952}
A.~L. Hodgkin and A.~F. Huxley.
\newblock A quantitative description of membrane current and its application to
  conduction and excitation in nerve.
\newblock {\em The Journal of Physiology}, 117(4):500--544, 1952.

\bibitem{hodgkin_measurement_1952}
A.~L. Hodgkin, A.~F. Huxley, and B.~Katz.
\newblock Measurement of current-voltage relations in the membrane of the giant
  axon of {Loligo}.
\newblock {\em The Journal of Physiology}, 116(4):424--448, April 1952.

\bibitem{huys_efficient_2006}
Quentin J.~M. Huys, Misha~B. Ahrens, and Liam Paninski.
\newblock Efficient {Estimation} of {Detailed} {Single}-{Neuron} {Models}.
\newblock {\em Journal of Neurophysiology}, 96(2):872--890, August 2006.

\bibitem{izhikevich_dynamical_2007}
Eugene~M. Izhikevich.
\newblock {\em Dynamical {Systems} in {Neuroscience}}.
\newblock MIT Press, Cambridge, MA, 2007.

\bibitem{jouffroy_tutorial_2010}
Jerome Jouffroy and Thor~I. Fossen.
\newblock A {Tutorial} on {Incremental} {Stability} {Analysis} using
  {Contraction} {Theory}.
\newblock 31(3):93--106, July 2010.

\bibitem{knopfel_optical_2019}
Thomas Knöpfel and Chenchen Song.
\newblock Optical voltage imaging in neurons: moving from technology
  development to practical tool.
\newblock {\em Nature Reviews Neuroscience}, 20(12):719--727, December 2019.

\bibitem{krstic_nonlinear_1995}
Miroslav Krstic, Petar~V. Kokotovic, and Ioannis Kanellakopoulos.
\newblock {\em Nonlinear and {Adaptive} {Control} {Design}}.
\newblock John Wiley \& Sons, Inc., New York, 1st edition, 1995.

\bibitem{ljung_convergence_1978}
L.~Ljung.
\newblock Convergence analysis of parametric identification methods.
\newblock {\em IEEE Transactions on Automatic Control}, 23(5):770--783, October
  1978.

\bibitem{ljung_system_1999}
Lennart Ljung.
\newblock {\em System {Identification}: {Theory} for the {User}}.
\newblock Prentice Hall PTR, Upper Saddle River, NJ, 1999.

\bibitem{lohmiller_contraction_1998}
Winfried Lohmiller and Jean-Jacques~E. Slotine.
\newblock On {Contraction} {Analysis} for {Non}-linear {Systems}.
\newblock {\em Automatica}, 34(6):683--696, June 1998.

\bibitem{lopez_adaptive_2021}
Brett~T. Lopez and Jean-Jacques~E. Slotine.
\newblock Adaptive {Nonlinear} {Control} {With} {Contraction} {Metrics}.
\newblock {\em IEEE Control Systems Letters}, 5(1):205--210, January 2021.

\bibitem{manchester_identification_2011-1}
I.~R. Manchester, M.~M. Tobenkin, and J.~Wang.
\newblock Identification of nonlinear systems with stable oscillations.
\newblock In {\em 2011 50th {IEEE} {Conference} on {Decision} and {Control} and
  {European} {Control} {Conference}}, pages 5792--5797, Orlando, FL, December
  2011.

\bibitem{marder_central_2001}
Eve Marder and Dirk Bucher.
\newblock Central pattern generators and the control of rhythmic movements.
\newblock {\em Current Biology}, 11(23):R986--R996, November 2001.

\bibitem{marder_neuromodulation_2014}
Eve Marder, Timothy O'Leary, and Sonal Shruti.
\newblock Neuromodulation of {Circuits} with {Variable} {Parameters}: {Single}
  {Neurons} and {Small} {Circuits} {Reveal} {Principles} of {State}-{Dependent}
  and {Robust} {Neuromodulation}.
\newblock {\em Annual Review of Neuroscience}, 37(1):329--346, 2014.

\bibitem{marino_global_1992}
R.~Marino and P.~Tomei.
\newblock Global adaptive observers for nonlinear systems via filtered
  transformations.
\newblock {\em IEEE Transactions on Automatic Control}, 37(8):1239--1245,
  August 1992.
\newblock Conference Name: IEEE Transactions on Automatic Control.

\bibitem{meliza_estimating_2014}
C.~Daniel Meliza, Mark Kostuk, Hao Huang, Alain Nogaret, Daniel Margoliash, and
  Henry D.~I. Abarbanel.
\newblock Estimating parameters and predicting membrane voltages with
  conductance-based neuron models.
\newblock {\em Biological Cybernetics}, 108(4):495--516, August 2014.

\bibitem{narayanan_biophysically_2019}
Vignesh Narayanan, Jr-Shin Li, and ShiNung Ching.
\newblock Biophysically interpretable inference of single neuron dynamics.
\newblock {\em Journal of Computational Neuroscience}, 47(1):61--76, August
  2019.

\bibitem{nicolas-alonso_brain_2012}
Luis~Fernando Nicolas-Alonso and Jaime Gomez-Gil.
\newblock Brain {Computer} {Interfaces}, a {Review}.
\newblock {\em Sensors}, 12(2):1211--1279, January 2012.

\bibitem{nogaret_automatic_2016}
Alain Nogaret, C.~Daniel Meliza, Daniel Margoliash, and Henry D.~I. Abarbanel.
\newblock Automatic {Construction} of {Predictive} {Neuron} {Models} through
  {Large} {Scale} {Assimilation} of {Electrophysiological} {Data}.
\newblock {\em Scientific Reports}, 6:32749, September 2016.

\bibitem{pascanu_difficulty_2013}
Razvan Pascanu, Tomas Mikolov, and Yoshua Bengio.
\newblock On the difficulty of training recurrent neural networks.
\newblock In {\em Proceedings of the 30th {International} {Conference} on
  {Machine} {Learning}}, volume~28 of {\em {ICML}'13}, pages 1310--1318,
  Atlanta, GA, USA, June 2013.

\bibitem{ribeiro_smoothness_2020}
Antônio~H. Ribeiro, Koen Tiels, Jack Umenberger, Thomas~B. Schön, and Luis~A.
  Aguirre.
\newblock On the smoothness of nonlinear system identification.
\newblock {\em Automatica}, 121:109158, November 2020.

\bibitem{sastry_adaptive_2011}
Shankar Sastry and Marc Bodson.
\newblock {\em Adaptive control: stability, convergence, and robustness}.
\newblock Dover Publications, Mineola, NY, 2011.

\bibitem{schmetterling_adaptive_2022}
Raphael Schmetterling, Thiago~B. Burghi, and Rodolphe Sepulchre.
\newblock Adaptive conductance control.
\newblock {\em Annual Reviews in Control}, August 2022.

\bibitem{sepulchre_control_2019-1}
R.~Sepulchre, G.~Drion, and A.~Franci.
\newblock Control {Across} {Scales} by {Positive} and {Negative} {Feedback}.
\newblock {\em Annual Review of Control, Robotics, and Autonomous Systems},
  2(1):89--113, 2019.

\bibitem{sharp_dynamic_1993}
A.~A. Sharp, M.~B. O'Neil, L.~F. Abbott, and E.~Marder.
\newblock Dynamic clamp: computer-generated conductances in real neurons.
\newblock {\em Journal of Neurophysiology}, 69(3):992--995, March 1993.

\bibitem{sorrell_brainmachine_2021}
Ethan Sorrell, Michael~E. Rule, and Timothy O’Leary.
\newblock Brain–{Machine} {Interfaces}: {Closed}-{Loop} {Control} in an
  {Adaptive} {System}.
\newblock {\em Annual Review of Control, Robotics, and Autonomous Systems},
  4(1):167--189, May 2021.

\bibitem{astrom_adaptive_2008}
Karl~Johan Åström and Björn Wittenmark.
\newblock {\em Adaptive {Control}}.
\newblock Dover Publications, Mineola, NY, 2nd edition, January 2008.

\bibitem{tyukin_adaptation_2007}
I.~Y. Tyukin, D.~V. Prokhorov, and C.~van Leeuwen.
\newblock Adaptation and {Parameter} {Estimation} in {Systems} {With}
  {Unstable} {Target} {Dynamics} and {Nonlinear} {Parametrization}.
\newblock {\em IEEE Transactions on Automatic Control}, 52(9):1543--1559,
  September 2007.

\bibitem{tyukin_adaptive_2013}
Ivan~Y. Tyukin, Erik Steur, Henk Nijmeijer, and Cees van Leeuwen.
\newblock Adaptive observers and parameter estimation for a class of systems
  nonlinear in the parameters.
\newblock {\em Automatica}, 49(8):2409--2423, August 2013.

\bibitem{yuste_cortex_2005}
Rafael Yuste, Jason~N. MacLean, Jeffrey Smith, and Anders Lansner.
\newblock The cortex as a central pattern generator.
\newblock {\em Nature Reviews Neuroscience}, 6(6):477--483, June 2005.

\bibitem{zhang_adaptive_2001}
Qinghua Zhang and A.~Clavel.
\newblock Adaptive observer with exponential forgetting factor for linear time
  varying systems.
\newblock In {\em 40th {IEEE} {Conference} on {Decision} and {Control}}, pages
  3886--3891, Orlando, FL, USA, December 2001.

\end{thebibliography}

\begin{IEEEbiography}[{\includegraphics[width=1in,height=1.25in,clip,keepaspectratio]
{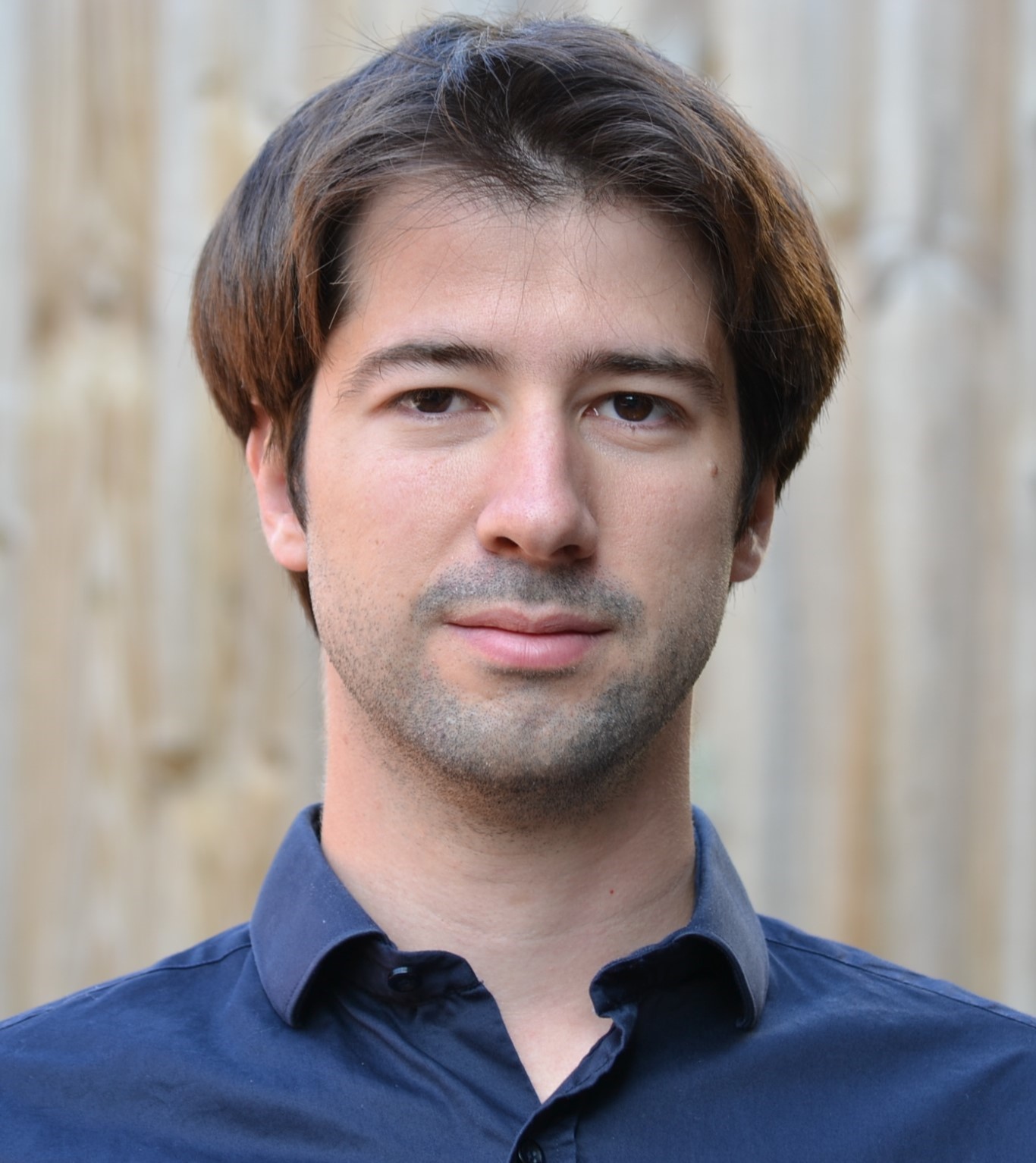}}]
{Thiago B. Burghi} 
Is a post-doctoral researcher in the Control Group of the Department of
Engineering at the University of Cambridge. He received the Diplôme d'Ingénieur
from ENSTA ParisTech, France, in 2012, and the B.Sc and M.Sc in Control
Engineering and Mechanical Engineering, respectively, from the University of
Campinas, Brazil, in 2015. He also holds a M.Sc in Robotics from the Pierre and
Marie Curie University (Paris VI), France. Thiago was awarded the Capes-Cambridge
Trust Scholarship to study at the University of Cambridge, UK, where he completed
his Ph.D in 2020. His research interests lie at the interface between nonlinear
control theory, system identification, and biophysical neuronal systems.
\end{IEEEbiography}

\begin{IEEEbiography}[{\includegraphics[width=1in,height=1.25in,clip,keepaspectratio]
{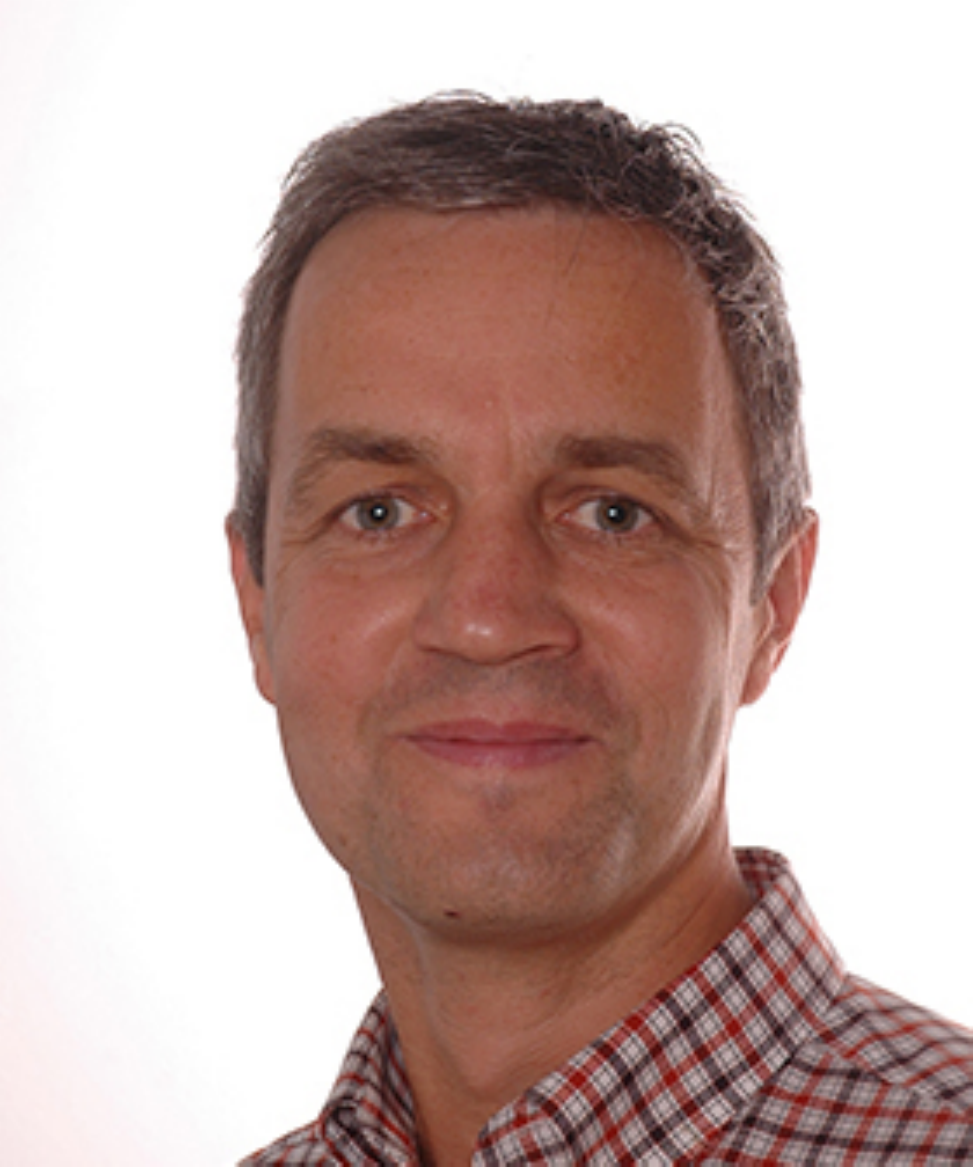}}]
{Rodolphe Sepulchre} 
(M96,SM08,F10) received the engineering degree and the Ph.D. degree from the
Université Catholique de Louvain in 1990 and in 1994, respectively.  He is
Professor of Engineering at the University of Cambridge since 2013.  His research
interests are in nonlinear control and optimization, and more recently
neuromorphic control.  He co-authored the monographs ``Constructive Nonlinear
Control'' (Springer-Verlag, 1997) and ``Optimization on Matrix Manifolds''
(Princeton University Press, 2008). He is Editor-in-Chief of IEEE Control Systems.
He is a recipient of  the IEEE CSS Antonio Ruberti Young Researcher Prize (2008)
and of the IEEE CSS George S. Axelby Outstanding Paper Award (2020). He is a
fellow of IEEE, IFAC,  and SIAM. He has been IEEE CSS distinguished lecturer 
between 2010 and 2015. In 2013, he was elected at the Royal Academy of Belgium.
\end{IEEEbiography}

\end{document}